%% file: main.tex
\documentclass[11pt]{article}
\usepackage{amssymb, amsmath, amsthm}

\usepackage{fullpage}
\usepackage[english]{babel} 
\usepackage[utf8]{inputenc} 
\usepackage{hyperref}
\usepackage{todonotes}
\usepackage{cleveref}
\usepackage{caption}
\usepackage{subcaption}
\usepackage{changes}
\usepackage{comment}
\usepackage{tabu}
\usepackage{multirow}
\usepackage{xcolor}

\usepackage{colortbl}
\usepackage{url}
\usepackage{float}
\usepackage{enumerate,paralist}
\usepackage{accents}
\usepackage{setspace}

\usepackage{algpseudocode}
\usepackage{algorithm}

\usepackage{booktabs}

\newcommand{\R}{\mathbb{R}}
\newcommand{\N}{\mathbb{N}}

\newcommand{\E}{\mathbb{E}}
\newcommand{\hhoracle}{\operatorname{\textsc{HH-Oracle}}}
\newcommand{\sketchalg}{\operatorname{\textsc{SketchAlg}}}
\newcommand{\counter}{\mathsf{counter}}

\newcommand{\hDS}{\mathcal{C}}
\newcommand{\QDist}{\mathcal{D}}

\DeclareMathOperator*{\median}{\mathsf{median}}

\newcommand{\eps}{\varepsilon}
\newcommand{\mydef}{:=}
\newcommand{\Var}{\mathrm{Var}}
\newcommand{\Err}{\mathrm{Err}}
\newcommand{\sF}{\mathcal{F}}
\newcommand{\sA}{\mathcal{A}}
\newcommand{\sD}{\mathcal{D}}

\newtheorem{theorem}{Theorem}[section]
\newtheorem{lemma}[theorem]{Lemma}

\newtheorem{corollary}[theorem]{Corollary}

\newtheorem{claim}[theorem]{Claim}

\theoremstyle{definition}

\newtheorem{remark}[theorem]{Remark}

\title{(Learned) Frequency Estimation Algorithms under Zipfian Distribution}
\date{}
\author{Anders Aamand%
	\thanks{BARC, University of Copenhagen, \protect\url{aa@di.ku.dk}}
\and Piotr Indyk%
	\thanks{CSAIL, MIT, \protect\url{indyk@mit.edu}} 
\and Ali Vakilian\thanks{University of Wisconsin-Madison, \protect\url{vakilian@wisc.edu}}}

\begin{document}
\maketitle

\input{abstract2}

\pagenumbering{gobble}
\newpage
\setcounter{page}{1}
\pagenumbering{arabic}

\input{intro}

\input{prelim}
\input{countmin}

\input{countsketch}

\input{learned-countsketch}

\input{experiments}

\bibliographystyle{alpha}
\bibliography{lib}
\appendix
\input{appendix}

\input{appendix2}
\end{document}

%% file: abstract2.tex
\begin{abstract}
The frequencies of the elements in a data stream are an important statistical measure and the task of estimating them arises in many applications within data analysis and machine learning. Two of the most popular algorithms for this problem, Count-Min and Count-Sketch, are widely used in practice. 

In a recent work [Hsu et al., ICLR'19], it was shown empirically that augmenting Count-Min and Count-Sketch with a machine learning algorithm leads to a significant reduction of the estimation error. The experiments were complemented with an analysis of the expected error incurred by Count-Min (both the standard and the augmented version) when the input frequencies follow a Zipfian distribution. Although the authors established that the  learned version of Count-Min has lower estimation error than its standard counterpart, their analysis  of the standard  Count-Min algorithm was not tight. Moreover, they provided no similar analysis for Count-Sketch. 

In this paper we resolve these problems. First, we provide a simple tight analysis of the expected error incurred by Count-Min. Second, we provide the first error bounds for both the standard and  the augmented version of Count-Sketch. These bounds are nearly tight and again demonstrate an improved performance of the learned version of Count-Sketch. 

In addition to demonstrating tight gaps between the aforementioned algorithms, we believe that our bounds for the standard versions of Count-Min and Count-Sketch are of independent interest. In particular, it is a typical practice to set the number of hash functions in those algorithms to $\Theta (\log n)$. In contrast, our results show that to minimize the \emph{expected} error, the number of hash functions should be a constant, strictly greater than $1$. 
\end{abstract}

%% file: intro.tex
\section{Introduction}

The last few  years have  witnessed  a rapid  growth in using machine learning methods to solve ``classical''  algorithmic problems.  For example, they have been used to improve the performance of
data structures~\cite{kraska2017case, mitz2018model}, 
online algorithms~\cite{lykouris2018competitive,purohit2018improving,pmlr-v97-gollapudi19a,kodialam2019optimal,chawla2019learning,angelopoulos2020online,lattanzi2020online,rohatgi2020near,antoniadis2020online}, 
combinatorial optimization~\cite{khalil2017learning,balcan2018learning,mitzenmacher2020scheduling}, 
similarity search~\cite{wang2016learning,dong2019learning}, 
compressive sensing~\cite{mousavi2015deep,bora2017compressed} and 
streaming algorithms~\cite{hsu2018learningbased,indyk2019learning,jiang2020learningaugmented,cohen2020composable}.
Multiple frameworks for designing and analyzing such algorithms have been proposed \cite{ailon2011self,gupta2017pac,balcan2018dispersion,alabi2019learning}.
The rationale behind this line of research is that machine learning makes it possible to adapt the behavior of the algorithms to inputs from a specific data distribution, making them more efficient or more accurate in specific applications. 



In this paper we focus on learning-augmented streaming algorithms for frequency estimation. The latter problem is formalized as follows: given a sequence $S$ of elements from some universe $U$,   construct a data structure that for any element $i \in U$ computes an estimation $\tilde{f}_i$ of  $f_i$, the number of times $i$ occurs in $S$. Since counting data  elements is a very common subroutine, frequency estimation algorithms have found applications in many areas, such as machine learning, network measurements and computer security. Many of the most popular algorithms for this problem, such as  Count-Min  (CM)~\cite{cormode2005improved}  or Count-Sketch (CS)~\cite{charikar2002finding} are based on hashing.  Specifically, these algorithms hash stream elements into $B$ buckets, count the number of items hashed into each bucket, and use the bucket value as an estimate of item frequency. To improve the accuracy, the algorithms use $k>1$ such hash functions and  aggregate the answers. These  algorithms have several useful properties:  they can handle item  deletions (implemented by decrementing the respective counters), and some of them (Count-Min)  never underestimate the true frequencies, i.e., $\tilde{f}_i \ge f_i$.


In a recent work~\cite{hsu2018learningbased}, the  authors showed that the aforementioned algorithm can be improved by augmenting them with machine learning. Their approach  is as follows.  During  the  training phase, they construct a classifier (neural network) to detect whether an element is ``heavy'' (e.g., whether $f_i$ is among top $k$ frequent items). After such a classifier is trained, they scan  the  input stream, and apply the  classifier to each element $i$. If  the element  is  predicted to be heavy,   it is allocated a unique bucket, so that an  exact value of $f_i$ is computed. Otherwise, the element is forwarded to a ``standard'' hashing data structure $\hDS$, e.g., CM or CS. To  estimate  $\tilde{f}_i$, the algorithm either returns  the exact count  $f_i$ (if $i$ is allocated a unique bucket) or an estimate provided by the data structure $\hDS$.\footnote{See Figure~\ref{fig:alg} for a generic implementation of the learning-based algorithms of~\cite{hsu2018learningbased}.}   
An empirical evaluation, on networking and query log data sets, shows that this approach can reduce  the overall estimation error. 


The  paper also presents a preliminary analysis of the algorithm. Under the common assumption  that the frequencies follow the Zipfian law, i.e.,\footnote{In fact we will assume that $f_i=1/i$. This is just a matter of scaling and is convenient as it removes the dependence of the length of the stream in our bounds} $f_i \propto 1/i$, for $i=1,\dots,n$ for some $n$, and further that item $i$ is queried with probability proportional to its frequency, the expected error incurred by the learning-augmented version of CM is shown  to be asymptotically lower than that of the ``standard'' CM.\footnote{This  assumes that the error rate for the ``heaviness'' predictor is sufficiently low.} However, the exact magnitude of the gap between the error incurred by the  learned and standard CM algorithms was left as an open problem. Specifically, \cite{hsu2018learningbased} only shows that the expected error of  standard CM with $k$ hash functions and a total of $B$ buckets is between 
$\frac{k }{B  \log(k)}$ 
and 
$\frac{k \log^{(k+2)/(k-1)} (kn/B) }{B}$.
Furthermore, no such analysis was presented for CS.


\subsection{Our results}\label{sec:ourresults} In this paper we resolve the aforementioned questions left open in~\cite{hsu2018learningbased}. Assuming that the frequencies follow a Zipfian law, we show:
\begin{itemize}
\item An  asymptotically  tight bound of $\Theta(\frac{ k\log(kn/B)}{B})$ for the expected error incurred by the CM algorithm with $k$ hash functions and a total of $B$ buckets. Together with a prior bound for Learned CM (Table~\ref{tbl:results}), this shows that learning-augmentation improves the error of CM by  a factor of $\Theta( \log(n)/\log(n/B))$ if the heavy hitter oracle is perfect.
\item The first error bounds for CS and Learned CS (see Table~\ref{tbl:results}). In particular, we show that for Learned CS, a single hash function as in~\cite{hsu2018learningbased} leads to an asymptotically optimal error bound, improving over standard CS by a  factor of $\Theta( \log(n)/\log(n/B))$  (same as CM).
\end{itemize}

We highlight that our results are presented assuming that we use a \emph{total} of $B$ buckets. With $k$ hash functions, the range of each hash functions is therefore $[B/k]$. We make this assumption since we wish to compare the expected error incurred by the different sketches when the total sketch size is fixed. 
\begin{table}[h!]
\centering	
\resizebox{.85\textwidth}{!}{%
\renewcommand{\arraystretch}{1.5}
\begin{tabular}{l|l l} 
\toprule
& $k=1$ & $k > 1$\\
\midrule
\textsf{\textbf{Count-Min (CM)}} & $\Theta \left( \frac{\log n}{B} \right)$~\cite{hsu2018learningbased}	&	$\Theta \left(\frac{k \cdot \log({{kn}\over B})}{B} \right)$ \\
\textsf{\textbf{Learned Count-Min (L-CM)}} & $\Theta \left(\frac{\log^2({n\over B})}{B\log n} \right)$~\cite{hsu2018learningbased}	&	$\Omega \left( \frac{\log^2 (\frac{n}{B})}{B\log n} \right)$~\cite{hsu2018learningbased}\\
\textsf{\textbf{Count-Sketch (CS)}} & $\Theta \left( \frac{  \log B}{B} \right)$ & $\Omega \left(\frac{k^{1/2}  }{B\log k } \right)$ and $O \left(\frac{k^{1/2}  }{B} \right)$\\ 
\textsf{\textbf{Learned Count-Sketch (L-CS)}} & $\Theta \left( \frac{\log \frac{n}{B}}{B \log n} \right)$ & $\Omega \left( \frac{\log \frac{n}{B}}{B \log n} \right)$ \\ 
\bottomrule
\end{tabular}
}
\caption{This table summarizes our and previously known results on the expected frequency estimation error of Count-Min (CM), Count-Sketch (CS) and their learned variants (i.e., L-CM and L-CS) that use $k$ functions and overall space $k\times {B\over k}$ under Zipfian distribution.  For CS, we assume that $k$ is odd (so that the median of $k$ values is well defined). 
}\label{tbl:results}
\end{table}

For our results on L-CS in~\Cref{tbl:results} we initially assume that the heavy hitter oracle is \emph{perfect}, i.e., that it makes no mistakes when classifying the heavy items. This is unlikely to be the case in practice, so we complement the results with an analysis of L-CS when the heavy hitter oracle may err with probability at most $\delta$ on each item. As $\delta$ varies in $[0,1]$, we obtain a smooth trade-off between the performance of L-CS and its classic counterpart. Specifically, as long as $\delta=O(1/\log B)$, the bounds are as good as with a perfect heavy hitter oracle.

In addition to clarifying the gap between the learned and standard variants of popular frequency estimation algorithms, our results provide  interesting insights about the algorithms themselves. For example, for both  CM and CS, the number of hash  functions $k$  is often selected to be $\Theta(\log n)$, in order to guarantee that {\em every}  frequency is estimated up  to a certain error bound. In contrast, we show that if instead the goal is to bound the {\em expected} error, then setting $k$ to a constant (strictly greater than $1$) leads to the asymptotic optimal performance.  
We remark that the same phenomenon holds not only for a Zipfian query distribution but in fact for an arbitrary distribution on the queries (see Remark~\ref{remark:gen-query-distribution}). 

Let us make the above comparison with previous known bounds for CM and CS a bit more precise. With frequency vector $\textbf{f}$ and for an element $x$ in the stream, we denote by $\textbf{f}_{-x}^{(B)}$, the vector obtained by setting the entry corresponding to $x$ as well as the $B$ largest entries of $\textbf{f}$ to $0$. The classic technique for analysing CM and CS (see, e.g.,~\cite{charikar2002finding}) shows that using a single hash function and $B$ buckets, with probability $\Omega(1)$, the error when querying the frequency of an element $x$ is $O(\|\textbf{f}_{-x}^{(B)}\|_1/B)$ for CM and $O(\|\textbf{f}_{-x}^{(B)}\|_2/\sqrt{B})$ for CS.  By creating $O(\log(1/\delta))$ sketches and using the median trick, the error probability can then be reduced to $\delta$.  For the Zipfian distribution, these two bounds become $O(\log(n/B)/B)$ and $O(1/B)$ respectively, and to obtain them with high probability for all elements we require a sketch of size $\Omega(B \log n)$. Our results imply that to obtain similar bounds on the expected error, we only require a sketch of size $O(B)$ and a constant number of hash functions. The classic approach described above does not yield tight bounds on the expected errors of CM and CS when $k>1$ and to obtain our bounds we have to introduce new and quite different techniques as to be described in Section~\ref{sec:techniques}.

Our techniques are quite flexible. To illustrate this, we study the performance of the classic Count-Min algorithm with one and more hash functions, as well as its learned counterparts, in the case where the input follows the following more general Zipfian distribution with exponent $\alpha>0$. This distribution is defined by $f_i \propto 1/i^\alpha$ for $i\in [n]$. We present the precise results in~\Cref{tbl:results3} in~\Cref{appendix}.

In Section~\ref{sec:experiment}, we complement our theoretical bounds with empirical evaluation of standard and learned variants of Count-Min and Count-Sketch on a synthetic dataset, thus providing a sense of the constant factors of our asymptotic bounds.

\subsection{Related work}
The frequency estimation problem and the closely related heavy hitters problem are two of the most fundamental problems in the field of streaming algorithms~\cite{cormode2005improved, cormode2005summarizing, charikar2002finding, muthukrishnan2005data, cormode2008finding, cormode2010methods, berinde2010space, minton2014improved, braverman2016beating, larsen2016heavy, anderson2017high, braverman2017bptree, bhattacharyya2018optimal}. In addition to the aforementioned hashing-based algorithms (e.g.,~\cite{cormode2005improved,charikar2002finding}), multiple non-hashing algorithms were also proposed, e.g.,~\cite{misra1982finding,manku2002approximate,metwally2005efficient}. These algorithms often exhibit better accuracy/space tradeoffs, but do not posses many of the properties of hashing-based methods, such as the ability to handle deletions as well as insertions. 

Zipf law is a common modeling tool used to evaluate the performance of frequency estimation algorithms, and has been used in many papers in this area, including~\cite{manku2002approximate,metwally2005efficient,charikar2002finding}. In its general form it postulates that $f_i$ is proportional to $1/i^\alpha$ for some exponent parameter $\alpha>0$. In this paper we focus mostly on the ``original'' Zipf law where $\alpha=1$. We do, however, study Count-Min for more general values of $\alpha$ and the techniques introduced in this paper can be applied to other values of the exponent $\alpha$ for Count-Sketch as well.

\subsection{Our techniques}\label{sec:techniques}
Our main contribution is our analysis of the standard Count-Min and Count-Sketch algorithms for Zipfians with $k>1$ hash functions. Showing the improvement for the learned counterparts is relatively simple (for Count-Min it was already done in~\cite{hsu2018learningbased}). In both of these analyses we consider a fixed item $i$ and bound $\E[|f_i-\tilde f_i|]$ whereupon linearity of expectation leads to the desired results. 
In the following we assume that $f_j=1/j$ for each $j\in [n]$ and describe our techniques for bounding $\E[|f_i-\tilde f_i|]$ for each of the two algorithms. 
\subparagraph*{Count-Min.} 
With a single hash function and $B$ buckets it is easy to see that the head of the Zipfian distribution, namely the items of frequencies $(f_j)_{j\in [B]}$, contribute with $\log B/B$ to the expected error $\E[|f_i-\tilde f_i|]$, whereas the light items contribute with $\log (n/B)/B$. Our main observation is that with more hash functions the expected contribution from the heavy items drops to $1/B$ and so, the main contribution comes from the light items. To bound the expected contribution of the heavy items to the error $|f_i-\tilde f_i|$ we bound the probability that the contribution from these items is at least $t$, then integrate over $t$. The main observation is that if the error is at least $t$ then for each of the hash functions, either there exist $t/s$ items in $[B]$ hashing to the same bucket as $i$ or there is an item $j\neq i$ in $[B]$ of weight at most $s$ hashing to the same bucket as $i$. By a union bound, optimization over $s$, and some calculations, this gives the desired bound. The lower bound follows from simple concentration inequalities on the contribution of the tail. In contrast to the analysis from~\cite{hsu2018learningbased} which is technical and leads to suboptimal bounds, our analysis is short, simple, and yields completely tight bounds in terms of all of the parameters $k,n$ and $B$.


\subparagraph*{Count-Sketch.} Simply put, our main contribution is an improved understanding of the distribution of random variables of the form $S=\sum_{i=1}^n f_i \eta_i \sigma_i$. 
Here the $\eta_i\in\{0,1\}$ are i.i.d Bernouilli random variables and the $\sigma_i\in \{-1,1\}$ are independent Rademachers, that is, $\Pr[\eta_i=1]=\Pr[\eta_i=-1]=1/2$. 
Note that the counters used in CS are random variables having precisely this form. 
Usually such random variables are studied for the purpose of obtaining large deviation results. In contrast, in order to analyze CS, we are interested in a fine-grained picture of the distribution within a ``small'' interval $I$ around zero, say with $\Pr[S\in I]=1/2$. 
For example, when proving a lower bound on $\E[|f_i-\tilde f_i|]$, we must establish a certain \emph{anti-concentration} of $S$ around $0$. 
More precisely we find an interval $J\subset I$ centered at zero such that $\Pr[S\in J]=O(1/\sqrt{k})$. Combined with the fact that we use $k$ independent hash functions as well as properties of the median and the binomial distribution, this gives that $\E[|f_i-\tilde f_i|]=\Omega(|J|)$. 
Anti-concentration inequalities of this type are in general notoriously hard to obtain but it turns out that we can leverage the properties of the Zipfian distribution, specifically its heavy head. 
For our upper bounds on $\E[|f_i-\tilde f_i|]$ we need strong lower bounds on $\Pr[S\in J]$ for intervals $J\subset I$ centered at zero. 
Then using concentration inequalities we can bound the probability that half of the $k$ relevant counters are smaller (larger) than the lower (highter) endpoint of $J$, i.e., that the median does not lie in $J$. Again this requires a precise understanding of the distribution of $S$ within $I$.

\subsection{Structure of the paper}
In~\Cref{sec:prelim} we describe the algorithms Count-Min and Count-Sketch. We also formally define the estimation error  that we will study as well as the Zipfian distribution. In~\Cref{sec:countmin,sec:countsketch} we provide our analyses of the expected error of Count-Min and Count-Sketch. In~\Cref{sec:learnedcountsketch} we analyze the performance of learned Count-Sketch both when the heavy hitter oracle is perfect and when it may misclassify each item with probability at most $\delta$. In~\Cref{sec:experiment} we present our experiments. Finally, in~\Cref{appendix}, we analyse Count-Min for the generalized Zipfian distribution with exponent $\alpha>0$ both in the classic and learned case and prove matching lower bounds for the learned algorithms. 

%% file: prelim.tex
\section{Preliminaries}\label{sec:prelim}
We start out by describing the sketching algorithms Count-Min and Count-Sketch. Common to both of these algorithms is that we sketch a stream $S$ of elements coming from some universe $U$ of size $n$. For notational convenience we will assume that $U=[n]\mydef\{1,\dots,n\}$. If item $i$ occurs $f_i$ times then either algorithm outputs an estimate $\tilde f_i$ of $f_i$.

\subparagraph*{Count-Min.} We use $k$ independent and uniformly random hash functions $h_1,\dots,h_k:[n] \to [B]$. Letting $C$ be an array of size $[k] \times [B]$ we let $C[\ell,b]=\sum_{j\in [n]}[h_\ell(j)=b]f_j$. When querying $i\in [n]$ the algorithm returns $\tilde f_i=\min_{\ell \in [k]} C[\ell,h_\ell(i)]$. Note that we always have that $\tilde f_i \geq f_i$.

\subparagraph*{Count-Sketch.} We pick independent and uniformly random hash functions $h_1,\dots,h_k:[n] \to [B]$ and $s_1,\dots,s_k:[n] \to  \{-1,1\}$. Again we initialize an array $C$ of size $[k] \times [B]$ but now we let $C[\ell,b]=\sum_{j\in [n]}[h_\ell(j)=b]s_\ell(j)f_j$. When querying $i\in [n]$ the algorithm returns the estimate $\tilde f_i=\median_{\ell \in [k]} s_{\ell}(i) \cdot C[\ell,h_{\ell}(i)]$.
\begin{remark} 
The bounds presented in~\Cref{tbl:results} assumes that the hash functions have codomain $[B/k]$ and not $[B]$, i.e., that the \emph{total} number of buckets is $B$. In the proofs to follows we assume for notational ease that the hash functions take value in $[B]$ and the claimed bounds follows immediately by replacing $B$ by $B/k$. 
\end{remark} 

\subparagraph*{Estimation Error.} To measure and compare the overall accuracy of different frequency estimation algorithms, we will use the {\em expected} estimation error which is defined as follows: let $\sF = \{f_1, \cdots, f_n\}$ and $\tilde{\sF}_{\sA} = \{\tilde{f}_1, \cdots, \tilde{f}_n\}$ respectively denote the actual frequencies and the estimated frequencies obtained from algorithm $\sA$ of items in the input stream. We remark that when $\sA$ is clear from the context we denote $\tilde{\sF}_{\sA}$ as $\tilde{\sF}$. Then we define
\begin{align}\label{eq:expected-error-objective}
  	\Err(\sF, \tilde{\sF}_{\sA}) \mydef  \E_{i\sim \QDist} |f_i - \tilde{f}_i|,
\end{align}
where $\QDist$ denotes the {\em query distribution} of the items. Here, similar to previous work (e.g.,~\cite{roy2016augmented,hsu2018learningbased}), we assume that the query distribution $\QDist$ is the same as the frequency distribution of items in the stream, i.e., for any $i^*\in [n]$, $\Pr_{i\sim \QDist}[i = i^*] \propto f_{i^*}$ (more precisely, for any $i^*\in [n]$, $\Pr_{i\sim \QDist}[i = i^*] = f_{i^*}/ N$ where $N = \sum_{i\in [n]} f_i$ denotes the total sum of all frequencies in the stream). 

\begin{remark}\label{remark:gen-query-distribution}
As all upper/lower bounds in this paper are proved by bounding the expected error when estimating the frequency of a single item, $\E[|\tilde f_i-f_i|]$, then using linearity of expectation, in fact we obtain bounds for \emph{any} query distribution $(p_i)_{i\in [n]}$. 
\end{remark}


\subparagraph*{Zipfian Distribution.} 
In our analysis we assume that the frequency distribution of items follows Zipf's law. That is, if we sort the items according to their frequencies with no loss of generality assuming that $f_{1} \geq f_{2} \geq \cdots \geq f_{n}$, then for any $i\in [n]$, $f_{i} \propto {1/i}$. In fact, we shall assume that $f_{i}= 1/i$, which is just a matter of scaling, and which conveniently removes the dependence on the length of the stream in our bounds. 
Assuming that the query distribution is the same as the distribution of the frequencies of items in the input stream (i.e., $\Pr_{i\sim \QDist}[i^*] = f_{i^*}/ N = 1/(i^*\cdot H_n)$ where $H_n$ denotes the $n$-th harmonic number), we can write the expected error in~\cref{eq:expected-error-objective} as follows:
\begin{align}\label{eq:simplified-error}
\Err(\sF, \tilde{\sF}_{\sA}) = \E_{i \sim \sD}[|f_i - \tilde{f}_i|] = {1\over N} \cdot \sum_{i\in [n]} |\tilde{f}_i - f_i| \cdot f_i = {1\over H_n} \cdot \sum_{i\in [n]} |\tilde{f}_i - f_i| \cdot {1\over i} 
\end{align}
Throughout this paper, we present our results with respect to the objective function at the right hand side of~\cref{eq:simplified-error}, i.e., $(1/H_n)\cdot \sum_{i=1}^n {|\tilde{f}_i - f_i| \cdot f_i}$. However, it is easy to use our results to obtain bounds for any query distribution as stated in~\Cref{remark:gen-query-distribution}.

Later we shall study the generalized Zipfian distribution with exponent $\alpha>0$. Sorting the items according to their frequencies, $f_1,\geq f_2\geq \cdots \geq f_n$, it holds for any $i \in [n]$ that $f_i \propto 1/i^\alpha$. Again we present our result with respect to the objective function $\sum_{i=1}^n {|\tilde{f}_i - f_i| \cdot f_i}$.


\begin{figure}[!h]
\begin{minipage}{\textwidth}
	\begin{algorithm}[H]
	{\small
	\caption{Learning-Based Frequency Estimation}
		\begin{algorithmic}[1]
		\Procedure{LearnedSketch}{$B$, $B_h$, $\hhoracle$, $\sketchalg$}
		\For{each stream element $i$} 
		\If{$\hhoracle(i)=1$} \Comment{predicts whether $i$ is heavy (in top $B_h$- frequent items)}
			\If{a unique bucket is already assigned to item $i$}
				\State $\counter_i \leftarrow \counter_i +1$
			\Else
				\State {\bf allocate} a new unique bucket to item $i$ and $\counter_i \leftarrow 1$ 
			\EndIf
		\Else 
			\State {\bf feed} $i$ to $\sketchalg(B - B_h)$ \Comment{an instance of $\sketchalg$ with $B-B_h$ buckets}
		\EndIf
		\EndFor
		\EndProcedure		
		\end{algorithmic}
		}
	\end{algorithm}
\end{minipage}\hfill
\caption{A generic learning augmented algorithm for the frequency estimation problem. $\hhoracle$ denotes a given learned oracle for detecting whether the item is among the top $B_h$ frequent items of the stream and $\sketchalg$ is a given (sketching) algorithm (e.g., CM or CS) for the frequency estimation problem.} \label{fig:alg}
\end{figure}
\subparagraph*{Learning Augmented Sketching Algorithms for Frequency Estimation.} In this paper, following the approach of~\cite{hsu2018learningbased}, the {\em learned} variants of CM and CS are algorithms augmented with a machine learning based {\em heavy hitters} oracle. 
More precisely, we assume that the algorithm has access to an oracle $\hhoracle$ that predicts whether an item is ``heavy'' (i.e., is one of the $B_h$ most frequent items) or not. 
Then, the algorithm treats heavy and non-heavy items differently: (a) a unique bucket is allocated to each heavy item and their frequencies are computed with no error, (b) the rest of items are fed to the given (sketching) algorithm $\sketchalg$ using the remaining $B - B_h$ buckets and their frequency estimates are computed via $\sketchalg$ (see Figure~\ref{fig:alg}). 
We shall assume that  $B_h=\Theta(B-B_h)=\Theta(B)$, that is, we use asymptotically the same number of buckets for the heavy items as for the sketching of the light items. One justification for this assumption is that in any case we can increase both the number of buckets for heavy and light items to $B$ without affecting the overall asymptotic space usage.

Note that, in general the oracle $\hhoracle$ can make errors.  In our analysis we first obtain a theoretical understanding, by assuming that the oracle is perfect, i.e., the error rate is zero. We later complement this analysis, by studying the incurred error  when the oracle misclassifies each item with probability at most $\delta$. 

%% file: countmin.tex
\section{Tight Bounds for Count-Min with Zipfians}\label{sec:countmin}
For both Count-Min and Count-Sketch we aim at analyzing the expected value of the variable $\sum_{i\in [n]}f_i\cdot|\tilde f_i- f_i|$ where $f_i=1/i$ and $\tilde{f_i}$ is the estimate of $f_i$ output by the relevant sketching algorithm. 
Throughout this paper we use the following notation: For an event $E$ we denote by $[E]$ the random variable in $\{0,1\}$ which is $1$ if and only if $E$ occurs.
We begin by presenting our improved analysis of Count-Min with Zipfians. The main theorem is the following.
\begin{theorem}\label{thm:simplecm}
Let $n,B,k \in \N$ with $k\geq 2$ and $B\leq n/k$. Let further $h_1,\dots,h_k: [n] \to [B]$ be independent  and truly random hash functions. For $i \in [n]$ define the random variable $\tilde{f_i}=\min_{\ell \in [k]} \left( \sum_{j\in [n]} [h_{\ell}(j)=h_\ell(i)]f_j \right)$. For any $i\in [n]$ it holds that $ \E[|\tilde f_i- f_i|]=\Theta \left( \frac{\log \left( \frac{n}{B} \right)}{B} \right)$.
\end{theorem}
Replacing $B$ by $B/k$ in~\Cref{thm:simplecm} and using linearity of expectation we obtain the desired bound for Count-Min in the upper right hand side of~\Cref{tbl:results}. The natural assumption that $B\leq n/k$ simply says that the total number of buckets is upper bounded by the number of items. 

To prove~\Cref{thm:simplecm} we start with the following lemma which is a special case of the theorem.
\begin{lemma}\label{simpleanalysis}
Suppose that we are in the setting of~\Cref{thm:simplecm} and further that\footnote{In particular we dispose with the assumption that $B \leq n/k$.} $n=B$. Then 
\begin{align*}
\E[|\tilde{f_i}-f_i|]=O \left( \frac{1}{n} \right).
\end{align*}
\end{lemma}
\begin{proof}
It suffices to show the result when $k=2$ since adding more hash functions and corresponding tables only decreases the value of $|\tilde f_i- f_i|$. Define $Z_\ell=\sum_{j\in [n]\setminus\{i\}} [h_\ell(j)=h_\ell(i)]f_j$ for $\ell \in [2]$ and note that these variables are independent. For a given $t\geq 3/n$ we wish to upper bound $\Pr[Z_{\ell}\geq t]$. 
Let $s<t$ be such that $t/s$ is an integer, and note that if $Z_{\ell}\geq t$ then either of the following two events must hold:
\begin{enumerate}
\item[$E_1$:] There exists a $j\in [n]\setminus \{i\}$ with $f_j>s$ and $h_{\ell}(j)=h_{\ell}(i)$.
\item[$E_2$:] The set $\{j\in [n]\setminus \{i\}:h_{\ell}(j)=h_{\ell}(i)\}$ contains at least $t/s$ elements.
\end{enumerate}
To see this, suppose that $Z_\ell\geq t$ and that $E_1$ does not hold. Then 
$$
t\leq Z_\ell =\sum_{j\in [n]\setminus\{i\}} [h_\ell(j)=h_\ell(i)]f_j\leq s |\{j\in [n]\setminus \{i\}:h_{\ell}(j)=h_{\ell}(i)\}|,
$$
so it follows that $E_2$ holds.
By a union bound,
\begin{align*}
\Pr[Z_{\ell}\geq t]\leq \Pr[E_1]+\Pr[E_2]\leq  \frac{1}{ns}+\binom{n}{t/s}n^{-t/s}\leq  \frac{1}{ns}+\left( \frac{es}{t}\right)^{t/s}.
\end{align*}
Choosing $s=\Theta(\frac{t}{\log (tn)})$ such that $t/s$ is an integer, and using $t\geq {3\over n}$, a simple calculation yields that $\Pr[Z_{\ell}\geq t]=O\left( \frac{\log (tn)}{tn} \right)$. Note that  $|\tilde{f_i}-f_i|=\min (Z_1,Z_2)$. As $Z_1$ and $Z_2$ are independent, $\Pr[|\tilde{f_i}-f_i|\geq t]=O\left( \left(\frac{\log (tn)}{tn}\right)^2 \right)$, so 
\begin{align*}
\E[|\tilde{f_i}-f_i|]=\int_{0}^\infty \Pr[Z \geq t ] \, dt\leq \frac{3}{n}+O \left(\int_{3/n}^\infty \left(\frac{\log (tn)}{tn}\right)^2 \, dt \right)=O\left( \frac{1}{n} \right).
\end{align*}
\end{proof}
We can now prove the full statement of~\Cref{thm:simplecm}.
\begin{proof}[Proof of~\Cref{thm:simplecm}]
We start out by proving the upper bound. Let $N_1=[B]\setminus \{i\}$ and $N_2=[n]\setminus ([B] \cup \{i\})$. 
Let $b\in [k]$ be such that $\sum_{j\in N_1} f_j \cdot [h_b(j)=h_b(i)]$ is minimal. Note that $b$ is itself a random variable. We also define
\begin{align*}
Y_1&=\sum_{j\in N_1} f_j\cdot  [h_b(j)=h_b(i)], \text{ and } Y_2=\sum_{j\in N_2} f_j\cdot  [h_b(j)=h_b(i)].
\end{align*}
Then, $|\tilde f_i- f_i|\leq Y_1 +Y_2$. Using~\Cref{simpleanalysis}, we obtain that $\E[Y_1]=O(\frac{1}{B})$. For $Y_2$ we observe that
\begin{align*}
\E[Y_2\mid b]=\sum_{j\in N_2}  \frac{f_j}{B}= O\left( \frac{ \log \left( \frac{n}{B}\right)}{B}\right).
\end{align*}
We conclude that 
\begin{align*}
\E[|\tilde f_i- f_i|]\leq \E[Y_1]+\E[Y_2]=\E[Y_1]+\E[\E[Y_2 \mid b]] =O\left( \frac{ \log \left( \frac{n}{B}\right)}{B}\right).
\end{align*}
Next we prove the lower bound. We have already seen that the main contribution to the error comes from the tail of the distribution. As the tail of the distribution is relatively ``flat'' we can simply apply a concentration inequality to argue that with probability $\Omega(1)$, we have this asymptotic contribution for each of the $k$ hash functions. To be precise, for $j\in [n]$ and $\ell \in [k]$ we define $X_\ell^{(j)}=f_j \cdot \left([h_\ell(j)=h_\ell(i)]-\frac{1}{B} \right)$. Note that the variables $(X_\ell^{(j)})_{j\in [n]}$ are independent. We also define $S_\ell=\sum_{j\in N_2} X_\ell^{(j)}$ for $\ell\in [k]$. Observe that $|X_\ell^{(j)}|\leq f_j \leq \frac{1}{B}$ for $j \geq B$, $\E[X_\ell^{(j)}]=0$, and that
\begin{align*}
\Var[S_\ell]=\sum_{j\in N_2} f_j^2 \left(\frac{1}{B}-\frac{1}{B^2} \right)\leq  \frac{1}{B^2}.
\end{align*}
Applying Bennett's inequality(\Cref{thm:Bennett} of~\Cref{appendix2}), with $\sigma^2= \frac{1}{B^2}$ and $M=1/B$ thus gives that 
\begin{align*}
\Pr[S_\ell\leq -t]\leq \exp\left(-h\left(tB \right)\right).
\end{align*}
Defining $W_\ell=\sum_{j\in N_2} f_j \cdot [h_\ell(j)=h_\ell(i)]$ it holds that $\E[W_\ell]=\Theta \left( \frac{ \log \left( \frac{n}{B}\right)}{B}\right)$ and $S_\ell =W_\ell-\E[W_\ell]$, so putting $t=\E[W_\ell]/2$ in the inequality above we obtain that 
\begin{align*}
\Pr[W_\ell\leq \E[W_\ell]/2]=\Pr[S_\ell\leq -\E[W_\ell]/2]\leq \exp\left(-h\left(\Omega\left( \log \frac{n}{B} \right) \right) \right).
\end{align*}
Appealing to~\Cref{asymptotics} and using that $B\leq n/k$ the above bound becomes
\begin{align}\label{bound}
\Pr[W_\ell\leq \E[W_\ell]/2]
&\leq \exp \left(- \Omega \left( \log \frac{n}{B} \cdot \log \left(\log \frac{n}{B}+1 \right) \right) \right) \nonumber \\
&= \exp(-\Omega (\log k \cdot \log ( \log k +1) ))=k^{-\Omega (\log (\log k+1))}.
\end{align}
By the independence of the events $(W_\ell> E[W_\ell]/2)_{\ell\in [k]}$, we have that
\begin{align*}
\Pr\left[|\tilde f_i- f_i|\geq \frac{\E[W_\ell]}{2} \right]\geq (1-k^{-\Omega (\log (\log k+1))})^k=\Omega(1),
\end{align*}
and so $\E[|\tilde f_i- f_i|]=\Omega(\E[W_\ell])=\Omega\left(\frac{\log\left( \frac{n}{B}\right)}{B}\right)$, as desired.
\end{proof}
\begin{remark}\label{remark:lowindependence} We have stated~\Cref{thm:simplecm} for truly random hash functions but it suffices with $O(\log B)$-independent hashing to prove the upper bound. Indeed, the only step in which we require high independence is in the union bound in~\Cref{simpleanalysis} over the $\binom{n}{t/s}$ subsets of $[n]$ of size $t/s$. To optimize the bound we had to choose $s=t/\log (tn)$, so that $t/s=\log(tn)$. As we only need to consider values of $t$ with $t\leq \sum_{i=1}^n f_i=O(\log n)$, in fact $t/s=O(\log n)$ in our estimates. Finally, we applied~\Cref{simpleanalysis} with $n=B$ so it follows that $O(\log B)$-independence is enough to obtain our upper bound.
\end{remark}

%% file: countsketch.tex
\section{(Nearly) Tight Bounds for Count-Sketch with Zipfians}\label{sec:countsketch}
In this section we proceed to analyze Count-Sketch for Zipfians either using a single or more hash functions. We start with two simple lemmas which for certain frequencies $(f_i)_{i \in [n]}$ of the items in the stream can be used to obtain respectively good upper and lower bounds on $\E[|\tilde f_i-f_i|]$ in Count-Sketch with a single hash function. We will use these two lemmas both in our analysis of standard and learned Count-Sketch for Zipfians.

\begin{lemma}\label{explemma}
Let $w=(w_1,\dots,w_n)\in \R^n$, $\eta_1,\dots,\eta_n$ Bernoulli variables taking value $1$ with probability $p$, and $\sigma_1,\dots ,\sigma_n\in \{-1,1\}$ independent Rademachers, i.e., $\Pr[\sigma_i=1]=\Pr[\sigma_i=-1]=1/2$. Let $S=\sum_{i=1}^n w_i \eta_i \sigma_i$. Then, $\E[|S|]= O \left(\sqrt{p} \|w\|_2 \right)$.
\end{lemma}
\begin{proof}
Using that $\E[\sigma_i\sigma_j]=0$ for $i\neq j$ and Jensen's inequality $\E[|S|]^{2}\leq \E[S^2]=\E \left[\sum_{i=1}^n w_i^2\eta_i \right]=p \|w \|_2^2$, from which the result follows.
\end{proof}

\begin{lemma}\label{lowerexplemma}
Suppose that we are in the setting of~\Cref{explemma}. Let $I\subset [n]$ and let $w_I\in \R^n$ be defined by $(w_I)_i=[i\in I]\cdot w_i$. Then
\begin{align*}
\E[|S|]\geq \frac{1}{2}p\left(1-p \right)^{|I|-1} \|w_I\|_1.
\end{align*}
\end{lemma}
\begin{proof}
Let $J=[n]\setminus I$, $S_1=\sum_{i\in I} w_i \eta_i \sigma_i$, and $S_2=\sum_{i\in J} w_i \eta_i \sigma_i$. Let $E$ denote the event that $S_1$ and $S_2$ have the same sign or $S_2=0$. Then $\Pr[E]\geq 1/2$ by symmetry. For $i\in I$ we denote by $A_i$ the event that $\{j\in I: \eta_j\neq 0\}=\{i\}$. Then $\Pr[A_i]=p(1-p)^{|I|-1}$ and furthermore $A_i$ and $E$ are independent.  If $A_i\cap E$ occurs, then $|S|\geq |w_i|$ and as the events $(A_i \cap E)_{i\in I}$ are disjoint it thus follows that $\E[|S|]\geq \sum_{i\in I} \Pr[A_i\cap E] \cdot|w_i|\geq  \frac{1}{2}p\left(1-p \right)^{|I|-1} \|w_I\|_1$.
\end{proof}
With these tools in hand, we proceed to analyse Count-Sketch for Zipfians with one and more hash functions in the next two sections.
\subsection{One hash function}
By the same argument as in the discussion succeeding~\Cref{thm:simplecm}, the following theorem yields the desired result for a single hash function as presented in~\Cref{tbl:results}.
\begin{theorem}\label{unlearnedcs1}
Suppose that $B\leq n$ and let $h:[n] \to [B]$ and $s:[n] \to \{-1,1\}$ be truly random hash functions. Define the random variable $\tilde{f_i}=\sum_{j\in [n]} [h(j)=h(i)]s(j)f_j$ for $i \in [n]$. 
Then 
\begin{align*}
\E[|\tilde f_i-s(i)f_i|]=\Theta \left( \frac{ \log B}{B} \right).
\end{align*}
\end{theorem}

\begin{proof}
Let $i\in [n]$ be fixed. 
We start by defining $N_1=[B] \setminus\{i\}$ and $N_2=[n] \setminus ([B]\cup \{i\})$ and note that
\begin{align*}
|\tilde {f_i}-s(i)f_i|
 \leq \left|\sum_{j\in N_1}[h(j)=h(i)] s(j)f_j\right|+\left|\sum_{j\in N_2}[h(j)=h(i)] s(j)f_j\right|:=X_1+X_2.
\end{align*}
Using the triangle inequality $\E[X_1]\leq \frac{1}{B}\sum_{j \in N_1} f_j = O(\frac{\log B}{ B})$.
Also, by~\Cref{explemma}, $\E[X_2]=O\left( \frac{1}{B}\right)$ and combining the two bounds we obtain the desired upper bound.
For the lower bound we apply~\Cref{lowerexplemma} with $I=N_1$ concluding that
\begin{align*}
\E[|\tilde {f_i}-s(i)f_i|]\geq \frac{1}{2B} \left(1-\frac{1}{B} \right)^{|N_1|-1}\sum_{i\in N_1} f_i=\Omega \left(\frac{\log B}{B} \right).
\end{align*}

\end{proof}
\subsection{Multiple hash functions}
Let $k\in \N$ be odd. For a tuple $x=(x_1,\dots,x_k)\in \R^k$ we denote by $\median x$ the median of the entries of $x$. The following theorem immediately leads to the result on CS with $k\geq 3$ hash functions claimed in~\Cref{tbl:results}.

\begin{theorem}\label{unlearnedmultiCS}
Let $k\geq 3$ be odd, $n\geq kB$, and $h_1,\dots,h_k:[n] \to [B]$ and $s_1,\dots,s_k:[n] \to \{-1,1\}$ be truly random hash functions. Define~$\tilde{f_i}=\median_{\ell\in[k]}\left(\sum_{j\in [n]} [h_\ell(j)=h_\ell(i)]s_\ell(j)f_j\right)$ for $i \in [n]$. 
Assume that\footnote{This very mild assumption can probably be removed at the cost of a more technical proof. In our proof it can even be replaced by $k\leq B^{2-\eps}$ for any $\eps=\Omega(1)$.} $k\leq B$. 
Then 
\begin{align*}
\E[|\tilde {f_i}-s(i)f_i|]=\Omega\left(\frac{1}{B \sqrt k \log k} \right), \quad \text{and} \quad \E[|\tilde {f_i}-s(i)f_i|]=O\left(\frac{1}{B\sqrt{k}}\right)
\end{align*}
\end{theorem}
The assumption $n\geq kB$ simply says that the total number of buckets is upper bounded by the number of items. Again using linearity of expectation for the summation over $i\in [n]$ and replacing $B$ by $B/k$ we obtain the claimed upper and lower bounds of $\frac{\sqrt k}{B \log k}$ and $\frac{\sqrt k}{B}$ respectively. 
We note that even if the bounds above are only tight up to a factor of $\log k$ they still imply that it is asymptotically optimal to choose $k=O(1)$, e.g. $k=3$. To settle the correct asymptotic growth is thus of merely theoretical interest.

In proving the upper bound in~\Cref{unlearnedmultiCS}, we will use the following result by Minton and Price (Corollary 3.2 of~\cite{minton2014improved}) proved via an elegant application of the Fourier transform.
\begin{lemma}[Minton and Price~\cite{minton2014improved}]\label{lemma:minton}
Let $\{X_i: i \in [n]\}$ be independent symmetric random variables such that $\Pr[X_i=0]\geq 1/2$ for each $i$. Let $X=\sum_{i=1}^n X_i$ and $\sigma^2=\E[X^2]=\Var[X]$. For $\eps<1$ it holds that $\Pr[|X|<\eps \sigma]=\Omega(\eps)$
\end{lemma}

\begin{proof}[Proof of~\Cref{unlearnedmultiCS}]
If $B$ (and hence $k$) is a constant, then the results follow easily from~\Cref{explemma}, so in what follows we may assume that $B$ is larger than a sufficiently large constant. We subdivide the exposition into the proofs of the upper and lower bounds.

\paragraph{Upper bound} 
Define $N_1=[B] \setminus \{i\}$ and $N_2=[n]\setminus ([B]\cup \{i\})$. Let for $\ell \in [k]$, $X_1^{(\ell)}=\sum_{j\in N_1}[h_\ell(j)=h_\ell(i)]s_\ell(j)f_j$ and $X_2^{(\ell)}=\sum_{j\in N_2} [h_\ell(j)=h_\ell(i)]s_\ell(j)f_j$ and let $X^{(\ell)}=X_1^{(\ell)}+X_2^{(\ell)}$.

As the absolute error in Count-Sketch with one pair of hash functions $(h,s)$ is always upper bounded by the corresponding error in Count-Min with the single hash function $h$, we can use the bound in the proof of~\Cref{simpleanalysis} to conclude that $\Pr[|X_1^{(\ell)}|\geq t] = O(\frac{\log (tB)}{tB})$, 
when $t\geq 3/B$. Also $\Var[X_2^{(\ell)}] = (\frac{1}{B}-\frac{1}{B^2}) \sum_{j\in N_2}f_j^2\leq \frac{1}{B^2}$,
so by Bennett's inequality (\Cref{thm:Bennett}) with $M=1/B$ and $\sigma^2=1/B^2$ and~\Cref{asymptotics},
\begin{align*}
\Pr[|X_2^{(\ell)}|\geq t]\leq 2\exp \left(-h(tB)\right)\leq 2\exp \left(-\frac{1}{2} tB \log \left(tB+1\right)\right)=O \left( \frac{\log (tB)}{tB} \right),
\end{align*}
for $t\geq \frac{3}{B}$. It follows that for $t\geq 3/B$, 
\begin{align*}
\Pr[|X^{(\ell)}|\geq 2t]\leq \Pr[(|X_1^{(\ell)}|\geq t)] +\Pr(|X_2^{(\ell)}|\geq t)]=O \left( \frac{\log (tB)}{tB} \right).
\end{align*}
Let $C$ be the implicit constant in the $O$-notation above. If $|\tilde{f_i}-s(i)f_i|\geq 2t$, at least half of the values $(|X^{(\ell)}|)_{\ell \in [k]}$ are at least $2t$. For $t\geq 3/B$ it thus follows by a union bound that 
\begin{align}\label{eq:tbig}
\Pr[|\tilde{f_i}-s(i)f_i|\geq 2t ]\leq 2 \binom{k}{\lceil k/2 \rceil}\left( C\frac{\log (tB)}{tB} \right)^{\lceil k/2\rceil}\leq 2 \left( 4C\frac{\log (tB)}{tB} \right)^{\lceil k/2\rceil}.
\end{align}
If $\alpha =O(1)$ is chosen sufficiently large it thus holds that
\begin{align*}
\int_{\alpha/B}^\infty \Pr[|\tilde{f_i}-s(i)f_i|\geq t ] \, dt
&= 2\int_{\alpha/(2B)}^\infty \Pr[|\tilde{f_i}-s(i)f_i|\geq 2t ] \, dt \\
&\leq \frac{4}{B}\int_{\alpha/2}^\infty \left( 4C\frac{\log (t)}{t} \right)^{\lceil k/2\rceil} \, dt \\
&\leq \frac{1}{B2^{k}}\leq  \frac{1}{B\sqrt{k}}.
\end{align*}
Here the first inequality uses~\cref{eq:tbig} and a change of variable. The second inequality uses that $\left(4C \frac{\log t}{t}\right)^{\lceil k/2 \rceil}\leq (C'/t)^{2k/5}$ for some constant $C'$ followed by a calculation of the integral.
Now,
$$
\E[|\tilde{f_i}-s(i)f_i|]=\int_{0}^\infty \Pr[|\tilde{f_i}-s(i)f_i|\geq t ] \, dt,
$$
 so for our upper bound it therefore suffices to show that $\int_{0}^{\alpha/B} \Pr[|\tilde{f_i}-s(i)f_i|\geq t ] \, dt=O\left( \frac{1}{B \sqrt{k}}\right)$. For this we need the following claim:
\begin{claim}\label{intervalclaim}
Let $I\subset \R$ be the closed interval centered at the origin of length $2t$, i.e., $I=[-t,t]$. Suppose that $0< t\leq \frac{1}{2B}$. For $\ell \in [k]$, $\Pr[X^{(\ell)}\in I]= \Omega (tB)$.
\end{claim}
\begin{proof}
Note that $\Pr[X_1^{(\ell)}=0]\geq \Pr[\bigwedge_{j\in N_1} (h_{\ell}(j)\neq h_{\ell}(i))] = (1-{1\over B})^{N_1}=\Omega(1)$. 
Secondly $\Var[X_2^{(\ell)}]= (\frac{1}{B}-\frac{1}{B^2})\sum_{j \in N_2}f_j^2\leq \frac{1}{B^2}$.
Using that $X_1^{(\ell)}$ and $X_2^{(\ell)}$ are independent and~\Cref{lemma:minton} with $\sigma^2=\Var[X_2^{(\ell)}]$, it follows that $\Pr[X^{(\ell)}\in I]=\Omega \left(\Pr[X_2^{(\ell)}\in I] \right)=\Omega(tB)$.
\end{proof}
Let us now show how to use the claim to establish the desired upper bound. For this let $0<t \leq \frac{1}{2B}$ be fixed. If $|\tilde{f_i}-s(i)f_i|\geq t$, at least half of the values $(X^{(\ell)})_{\ell \in [k]}$ are at least $t$ or at most $-t$. Let us focus on bounding the probability that at least half are at least $t$, the other bound being symmetric giving an extra factor of $2$ in the probability bound. By symmetry and~\Cref{intervalclaim}, $\Pr[X^{(\ell)}\geq t]=\frac{1}{2}-\Omega(tB)$. For $\ell\in [k]$ we define $Y_{\ell}=[X^{(\ell)}\geq t]$, and we put $S=\sum_{\ell \in [k]}Y_\ell$. Then $\E[S]=k \left(\frac{1}{2}-\Omega(tB) \right)$. If at least half of the values $(X^{(\ell)})_{\ell \in [k]}$ are at least $t$ then $S\geq k/2$. By Hoeffding's inequality (\Cref{thm:Hoeffding}) we can bound the probability of this event by
\begin{align*}
\Pr[S\geq k/2]=\Pr[S-\E[S]=\Omega(ktB)]=\exp(-\Omega(kt^2B^2)).
\end{align*}
It follows that $\Pr[|\tilde{f_i}-s(i)f_i|\geq t]\leq 2\exp(-\Omega(kt^2B^2))$. Thus 
\begin{align*}
\int_{0}^{\alpha/B} \Pr[|\tilde{f_i}-s(i)f_i|\geq t ] \, dt 
&\leq \int_{0}^{\frac{1}{2B}} 2\exp(-\Omega(kt^2B^2)) \, dt+ \int_{\frac{1}{2B}}^{\alpha/B}2\exp(-\Omega(k)) \,dt \\
&\leq\frac{1}{B \sqrt{k}}\int_0^{\sqrt{k}/2} \exp(-t^2)\, dt+\frac{2\alpha \exp(-\Omega(k))}{B}=O\left(\frac{1}{B\sqrt{k}}\right).
\end{align*}
Here the second inequality used a change of variable. The proof of the upper bound is complete.

\paragraph{Lower Bound} 
Fix $\ell \in [k]$ and let $M_1=[B \log k]\setminus \{i\}$ and $M_2=[n]\setminus ([B \log k] \cup \{i\})$. Write
\begin{align*}
S:=\sum_{j\in M_1} [h_\ell(j)=h_\ell(i)]s_\ell(j)f_j
+\sum_{j\in M_2} [h_\ell(j)=h_\ell(i)]s_\ell(j)f_j 
:= S_1+S_2.
\end{align*}
We also define $J:=\{j\in M_1: h_\ell(j)=h_{\ell}(i) \}$. Let $I\subseteq \R$ be the closed interval around $s_\ell(i)f_i$ of length $\frac{1}{ B \sqrt{k}\log k}$. We now upper bound the probability that $S\in I$ conditioned on the value of $S_2$. To ease the notation, the conditioning on $S_2$ has been left out in the notation to follow. Note first that
\begin{align}\label{eq:somethingtobound}
\Pr[S\in I ]=\sum_{r=0}^{|M_1|} \Pr[S\in I \mid |J|=r] \cdot \Pr[|J|=r].
\end{align}
For a given $r\geq 1$ we now proceed to bound $\Pr[S\in I \mid |J|=r]$. This probability is the same as the probability that $S_2+\sum_{j\in R} \sigma_jf_j\in I$, where $R\subseteq M_1$ is a uniformly random $r$-subset and the $\sigma_j$'s are independent Rademachers. Suppose that we sample the elements from $R$ as well as the corresponding signs $(\sigma_i)_{i \in R}$ sequentially, and let us condition on the values and signs of the first $r-1$ sampled elements. At this point at most $\frac{B \log k}{\sqrt{k}}+1 $ possible samples for the last element in $R$ can cause that $S \in I$. Indeed, the minimum distance between distinct elements of $\{f_j: j \in M_1\}$ is at least $1/(B\log k)^2$ and furthermore $I$ has length $\frac{1}{B\sqrt{k} \log k}$. Thus, at most
\begin{align*}
 \frac{1}{ B \sqrt{k}\log k} \cdot (B\log k)^2+1 = \frac{B \log k}{\sqrt{k}}+1
\end{align*}
 choices for the last element of $R$ ensure that $S\in I$.  For $1\leq r \leq (B \log k) /2$ we can thus upper bound
\begin{align*}
\Pr[S\in I \mid |J|=r]\leq \frac{\frac{B \log k}{\sqrt{k}}+1}{|M_1| -r+1}\leq \frac{2}{\sqrt{k}}+\frac{2}{B \log k}\leq \frac{3}{\sqrt{k}}.
\end{align*}
Note that $\mu:=\E[|J|]\leq \log k$ so for $B\geq 6$, it holds that 
\begin{align*}
\Pr[|J| \geq  (B \log k) /2]\leq \Pr\left[|J|\geq \mu\frac{B}{2}\right]\leq \Pr\left[|J|\geq \mu \left(1+\frac{B}{3}\right)\right]\leq  \exp \left( -\mu h(B/3) \right)=k^{-\Omega (h(B/3))},
\end{align*}
where the last inequality follows from the Chernoff bound of~\Cref{thm:Chernoff}. Thus, if we assume that $B$ is larger than a sufficiently large constant, then $\Pr[|J|\geq B \log k /2]\leq k^{-1}$. Finally, $\Pr[|J|=0]=(1-1/B)^{B\log k}\leq k^{-1}$. Combining the above, we can continue the bound in~\eqref{eq:somethingtobound} as follows. 
\begin{align}
\Pr[S\in I ] 
\leq &  \Pr[|J|=0]+\sum_{r=1}^{(B \log k)/2}   \Pr[S\in I \mid |J|=r] \cdot \Pr[|J|=r] \nonumber\\
+&\sum_{r=(B \log k) /2+1}^{|M_1|} \Pr[|J|=r]= O \left(\frac{1}{\sqrt{k}} \right), \label{eq:SinI}
\end{align}
which holds even after removing the conditioning on $S_2$. We now show that with probability $\Omega (1)$ at least half the values $(X^{(\ell)})_{\ell \in [k]}$ are at least $\frac{1}{2B \sqrt k \log k}$. Let $p_0$ be the probability that $X^{(\ell)}\geq \frac{1}{2B \sqrt k \log k}$. This probability does not depend on $\ell \in [k]$ and by symmetry and~\eqref{eq:SinI}, $p_0=1/2-O(1/\sqrt{k})$. Define the function $f:\{0,\dots,k\} \to \R$ by
\begin{align*}
f(t)=\binom{k}{t} p_0^t (1-p_0)^{k-t}.
\end{align*}
Then $f(t)$ is the probability that exactly $t$ of the values $(X^{(\ell)})_{\ell \in [k]}$ are at least $\frac{1}{B \sqrt k \log k}$. Using that $p_0=1/2-O(1/\sqrt{k})$, a simple application of Stirling's formula gives that $f(t)=\Theta \left( \frac{1}{\sqrt{k}} \right)$ for $t=\lceil k/2\rceil,\dots,\lceil k/2+\sqrt{k}\rceil$ when $k$ is larger than some constant $C$. 
It follows that with probability $\Omega(1)$ at least half of the $(X^{(\ell)})_{\ell \in [k]}$ are at least $\frac{1}{B \sqrt k \log k}$ and in particular 
\begin{align*}
\E[|\tilde {f_i}-f_i|]=\Omega\left(\frac{1}{B \sqrt k \log k} \right).
\end{align*}
Finally we handle the case where $k\leq C$. It follows from simple calculations (e.g., using~\Cref{lowerexplemma}) that $X^{(\ell)}=\Omega(1/B)$ with probability $\Omega(1)$. Thus this happens for all $\ell \in [k]$ with probability $\Omega(1)$ and in particular $\E[|\tilde {f_i}-f_i|]= \Omega(1/B)$, which is the desired for constant $k$.
\end{proof}

%% file: learned-countsketch.tex
\section{Learned Count-Sketch for Zipfians}\label{sec:learnedcountsketch}
We now proceed to analyze the learned Count-Sketch algorithm. In~\Cref{sec:learned1} we estimate the expected error when using a single hash function and in~\Cref{sec:learnedk} we show that the expected error only increases when using more hash functions. Recall that we assume that the number of buckets $B_h$ used to store the heavy hitters that $B_h=\Theta(B-B_h)=\Theta(B)$. 
\subsection{One hash function}\label{sec:learned1}
By taking $B_1=B_h=\Theta(B)$ and $B_2=B-B_h=\Theta(B)$ in the theorem below, the result on L-CS for $k=1$ claimed in~\Cref{tbl:results} follows immediately. 

\begin{theorem}\label{thm:lcs1}
Let $h:[n]\setminus [B_1] \to [B_2]$ and $s:[n] \to \{-1,1\}$ be truly random hash functions where $n,B_1,B_2 \in \N$ and\footnote{The first inequality is the standard assumption that we have at least as many items as buckets. The second inequality says that we use at least as many buckets for non-heavy items as for heavy items (which doesn't change the asymptotic space usage).} $n-B_1\geq B_2\geq B_1$. Define the random variable $\tilde{f_i}=\sum_{j=B_1+1}^n [h(j)=h(i)]s(j)f_j$ for $i \in [n]\setminus [B_1]$. Then 
\begin{align*}
\E[|\tilde f_i-s(i)f_i|]=\Theta  \left(\frac{\log \frac{B_2+B_1}{B_1}}{B_2} \right)
\end{align*}
\end{theorem}
\begin{proof}
Let $N_1=[B_1+B_2]\setminus ([B_1]\cup \{i\})$ and $N_2=[n]\setminus ([B_1+B_2]\cup \{i\})$.
Let $X_1=\sum_{j\in N_1} [h(j)=h(i)]s(j)f_j$ and $X_2=\sum_{j\in N_2} [h(j)=h(i)]s(j)f_j$.
By the triangle inequality and linearity of expectation, 
\begin{align*}
\E[|X_1|]= O \left(\frac{\log \frac{B_2+B_1}{B_1}}{B_2} \right).
\end{align*}
Moreover, it follows directly from~\Cref{explemma} that  $\E\left[|X_2| \right]=O\left( \frac{1}{B_2} \right)$. Thus 
\begin{align*}
\E[|\tilde f_i-s(i)f_i|]\leq \E[|X_1|]+\E[|X_2|]=O \left(\frac{\log \frac{B_2+B_1}{B_1}}{B_2} \right),
\end{align*}
as desired. For the lower bound on $\E\left[\left|\tilde {f_i}-s(i)f_i \right| \right]$ we apply~\Cref{lowerexplemma} with $I=N_1$ to obtain that,
\begin{align*}
\E\left[\left|\tilde {f_i}-s(i)f_i \right| \right]\geq \frac{1}{2B_2} \left(1-\frac{1}{B_2} \right)^{|N_1|-1}\sum_{i\in N_1}f_i=\Omega \left(\frac{\log \frac{B_2+B_1}{B_1}}{B_2} \right).
\end{align*}
\end{proof}

\begin{corollary}\label{cro:lcs1}
Let $h:[n]\setminus [B_h] \to [B-B_h]$ and $s:[n] \to \{-1,1\}$ be truly random hash functions where $n,B,B_h \in \N$ and $B_h = \Theta(B) \leq B/2$. Define the random variable $\tilde{f_i}=\sum_{j=B_h+1}^n [h(j)=h(i)]s(j)f_j$ for $i \in [n]\setminus [B_h]$. Then $\E[|\tilde f_i-s(i)f_i|]=\Theta(1 / B)$.
\end{corollary}
\begin{remark}
The upper bounds of~\Cref{thm:lcs1} and~\Cref{cro:lcs1} hold even without the assumption of fully random hashing. In fact, we only require that $h$ and $s$ are $2$-independent. Indeed~\Cref{explemma} holds even when the Rademachers are $2$-independent (the proof is the same). Moreover, we need $h$ to be $2$-independent as we condition on $h(i)$ in our application of~\Cref{explemma}. With $2$-independence the variables $[h(j)=h(i)]$ for $j\neq i$ are then Bernoulli variables taking value $1$ with probability $1/B_2$.
\end{remark}


\subsection{More hash functions}\label{sec:learnedk}
We now show that, like for Count-Sketch, using more hash functions does not decrease the expected error. We first state the Littlewood-Offord lemma as strengthened by Erd\H{o}s.
\begin{theorem}[Littlewood-Offord~\cite{littlewood1939number}, Erd\H{o}s~\cite{erdos1945lemma}]\label{LittlewoodOfford} 
Let $a_1,\dots,a_n\in \R$ with $|a_i|\geq1$ for $i\in [n]$. Let further  $\sigma_1,\dots,\sigma_n\in \{-1,1\}$ be random variables with $\Pr[\sigma_i=1]=\Pr[\sigma_i=-1]=1/2$ and define $S=\sum_{i=1}^n \sigma_i a_i$. For any $v\in \R$ it holds that $\Pr[|S-v|\leq 1] = O(1/\sqrt{n})$.
\end{theorem}

Setting $B_1=B_h=\Theta(B)$ and $B_2=B-B_2=\Theta(B)$ in the theorem below gives the final bound from~\Cref{tbl:results} on L-CS with $k\geq 3$.
\begin{theorem}
Let $n\geq B_1+B_2\geq 2B_1$, $k\geq 3$ odd, and $h_1,\dots,h_k:[n] \setminus [B_1] \to [B_2/k]$ and $s_1,\dots,s_k:[n] \setminus [B_1]\to \{-1,1\}$ be independent and truly random. Define the random variable $\tilde{f_i}=\median_{\ell\in[k]}\left(\sum_{j\in [n]\setminus [B_1]} [h_\ell(j)=h_\ell(i)]s_\ell(j)f_j\right)$ for $i \in [n]\setminus [B_1]$. 
Then 
\begin{align*}
\E[|\tilde f_i-s(i)f_i|]=\Omega\left(\frac{1}{B_2}\right).
\end{align*}
\end{theorem}
\begin{proof}
Like in the proof of the lower bound of~\Cref{unlearnedmultiCS} it suffices to show that for each $i$ the probability that the sum $S_\ell:=\sum_{j\in [n]\setminus ([B_1]\cup \{i\})} [h_\ell(j)=h_\ell(i)]s_\ell(j)f_j$ lies in the interval $I=\left[-1/(2B_2),1/(2B_2) \right]$ is $O(1/\sqrt{k})$. Then at least half the $(S_\ell)_{\ell \in [k]}$ are at least $1/(2B_2)$ with probability $\Omega(1)$ by an application of Stirling's formula, and it follows that $\E[|\tilde f_i-s(i)f_i|]=\Omega (1/B_2)$. 

Let $\ell \in [k]$ be fixed, $N_1=[2B_2]\setminus ([B_2]\cup \{i\})$, and $N_2=[n] \setminus (N_1\cup \{i\})$, and write 
\begin{align*}
S_{\ell}=\sum_{j\in N_1}[h_\ell(j)=h_\ell(i)]s_\ell(j)f_j+\sum_{j\in N_2}[h_\ell(j)=h_\ell(i)]s_\ell(j)f_j:=X_1+X_2.
\end{align*}
Now condition on the value of $X_2$. Letting $J=\{j\in N_1:h_\ell(j)=h_\ell(i)\}$ it follows by~\Cref{LittlewoodOfford} that 
\begin{align*}
\Pr[S_\ell \in I \mid X_2]=O\left( \sum_{J'\subseteq N_1}\frac{\Pr[J=J']}{\sqrt{|J'|+1}} \right)=O\left( \Pr[|J|< k/2]+1/\sqrt{k}\right).
\end{align*}
An application of Chebyshev's inequality gives that $\Pr[|J|< k/2]=O(1/k)$, so $\Pr[S_\ell\in I]=O(1/\sqrt{k})$. Since this bound holds for any possible value of $X_2$ we may remove the conditioning and the desired result follows.
\end{proof}
\begin{remark} 
The bound above is probably only tight for $B_1=\Theta(B_2)$. Indeed, we know that it cannot be tight for all $B_1\leq B_2$ since when $B_1$ becomes very small, the bound from the standard Count-Sketch with $k\geq 3$ takes over --- and this is certainly worse than the bound in the theorem. It is an interesting open problem (that requires a better anti-concentration inequality than the Littlewood-Offord lemma) to settle the correct bound when $B_1 \ll B_2$. 
\end{remark}

\subsection{Learned Count-Sketch using a noisy heavy hitter oracle}
In~\cite{hsu2018learningbased} it was demonstrated that if the heavy hitter oracle is noisy, misclassifying an item with probability $\delta$, then the expected error incurred by Count-Min for Zipfians is 
$$
O \left(\frac{1}{\log n} \frac{\delta^2 \ln^2 B_h+\ln^2(n/B_h)}{B-B_h}\right).
$$
Here $B_h$ is the number of buckets used to store the heavy hitters and $B$ is the total number of buckets. Taking $B_h=\Theta(B)=\Theta(B-B_h)$, this bound becomes $O \left( \frac{\delta^2 \ln^2 B+\ln^2(n/B)}{B \log n}\right)$. As $\delta$ varies in $\left[0,1\right]$, this interpolates between the expected error incurred in respectively the learned case with a perfect heavy hitter oracle and the classic case. In particular it is enough to assume that $\delta =O(\ln(n/B)/\ln(B))$ in order to obtain the results in the idealized case with a perfect oracle.

We now provide a similar analysis for the learned Count-Sketch. More precisely we assume that we allocate $B_h$ buckets to the heavy hitters and $B-B_h$ to the lighter items. We moreover assume access to a heavy hitter oracle $\textbf{HH}_\delta$ such that for each $i\in [n]$, $\Pr[\textbf{HH}_\delta(i)\neq \textbf{HH}_0(i)]\leq \delta$, where $\textbf{HH}_0$ is a perfect heavy hitter oracle that correctly classifies the $B_h$ heaviest items.
\begin{theorem}\label{thm:faultyoracle}
Learned Count-Sketch with a single hash functions, a heavy hitter oracle $\textbf{HH}_\delta$, $B_h=\Theta(B)$ bins allocated to store the $B_h$ items classified as heavy and $B-B_h=\Theta(B)$ bins allocated to a Count-Sketch of the remaining items, incurs an expected error of 
$$
 O \left( \frac{(\delta \log B+\log (n/B))(1+\delta \log B)}{ B \log n} \right).
$$
\end{theorem}
\begin{proof}
Let $h:[n] \to [B-B_h]$ and $s:[n] \to \{-1,1\}$ be the hash functions used for the Count-Sketch. In the analysis to follow, it is enough to assume that they are $2$-independent. Suppose item $i$ is classified as non-heavy. For $j \in [n]$, let $\eta_j=[h(j)=h(i)]$, and let $\alpha_j$ be the indicator for item $j$ being classified as non-heavy. Then 
$$
|\tilde f_i-f_i|=\left| \sum_{j \in [n]\setminus \{i\}}\alpha_j \eta_j s(j) f_j\right| \leq \sum_{j \in [B_h] \setminus \{i\}} \alpha_j \eta_jf_j
+ \left| \sum_{j \in [n]\setminus (B_h \cup\{i\})}\alpha_j \eta_j s(j) f_j \right|:= S_1+S_2
$$
Note that $\E[S_1]=O\left(\frac{\delta \log B_h}{B-B_h} \right)=O\left(\frac{\delta \log B}{B} \right)$. For $S_2$, we let $p_j=\Pr[\alpha_j \eta_j=1]\leq \frac{1}{B-B_h}=O(\frac{1}{B})$. Then 
$$
\E[S_2] \leq (\E[S_2^2])^{1/2}=\left(\sum_{j \in [n]\setminus (B_h \cup\{i\})} p_jf_j^2 \right)^{1/2}=O \left(\frac{1}{B} \right),
$$
using that $\E[s(i) s(j)]=0$ for $i\neq j$ as $s$ is $2$-independent.
It follows that $\E[|\tilde f_i-f_i|]=O\left(\frac{1+\delta \log B}{B} \right)$, given that item $i$ is classified as non-heavy. Let $N=\sum_{i \in [n]}f_i=\Theta(\log n)$. As the probability of item $i\in [B_h]$ being classified as non-heavy is at most $\delta$, the the expected error is upper bounded by
$$
\frac{1}{N} \left (\delta \sum_{j \in [B_h]\setminus \{i\}} f_i +\sum_{j \in [n]\setminus (B_h \cup\{i\})} f_i\right) \cdot O\left(\frac{1+\delta \log B}{B} \right)= O \left( \frac{(\delta \log B+\log (n/B))(1+\delta \log B)}{ B \log n} \right),
$$
as desired.
\end{proof}
We see that with $\delta=1$, we recover the bound of $\frac{\log B}{B}$ presented in~\Cref{tbl:results} for the classic Count-Sketch. On the other hand, it is enough to assume that $\delta=O(1/\log B)$ in order to obtain the bound of $O \left( \frac{\log (n/B)}{B \log n} \right)$, which is what we obtain with a perfect heavy hitter oracle.

%% file: experiments.tex
\section{Experiments}\label{sec:experiment}
In this section, we provide the empirical evaluation of CountMin, CountSketch and their learned counterparts under Zipfian distribution. Our empirical results complement the theoretical analysis provided earlier in this paper.
\subparagraph*{Experiment setup.} We consider a {\em synthetic} stream of $n = 10K$ items where the frequencies of the items follow the standard Zipfian distribution (i.e., with $\alpha=1$). To be consistent with our assumption in our theoretical analysis, we scale the frequencies so that the frequency of item $i$ is $1/i$. In our experiments, we vary the values of the number of buckets ($B$) and the number of rows in the sketch ($k$) as well as the number of predicted heavy items in the learned sketches. We remark that in this section we assume that the heavy hitter oracle predicts {\em without errors}.

We run each experiment 20 times and take the average of the estimation error defined in eq.~\eqref{eq:simplified-error}.

\begin{figure*}[!h]
\minipage{0.5\textwidth}
		\includegraphics[width=\textwidth]{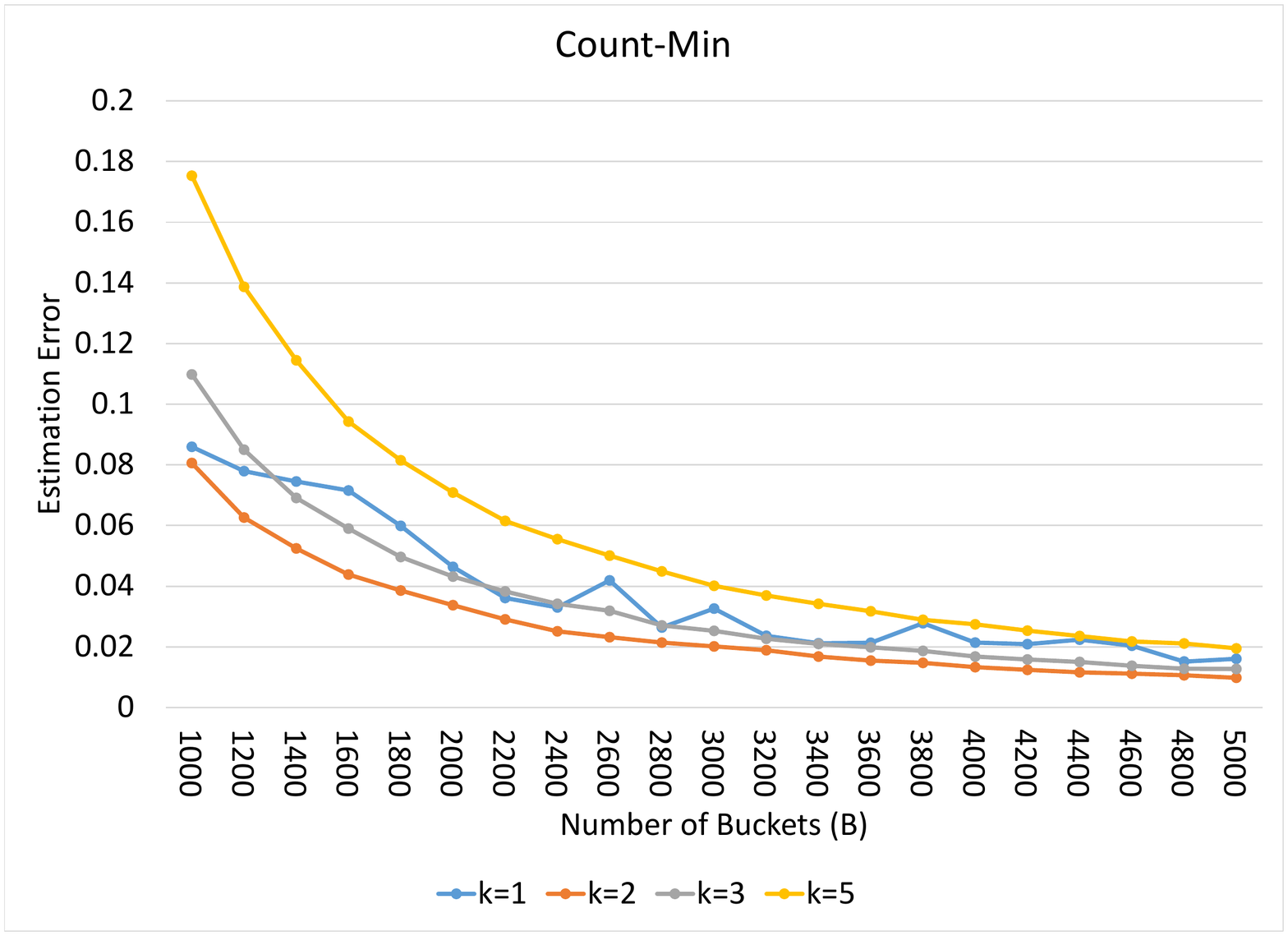}
\endminipage\hfill
\minipage{0.5\textwidth}
		\includegraphics[width=\textwidth]{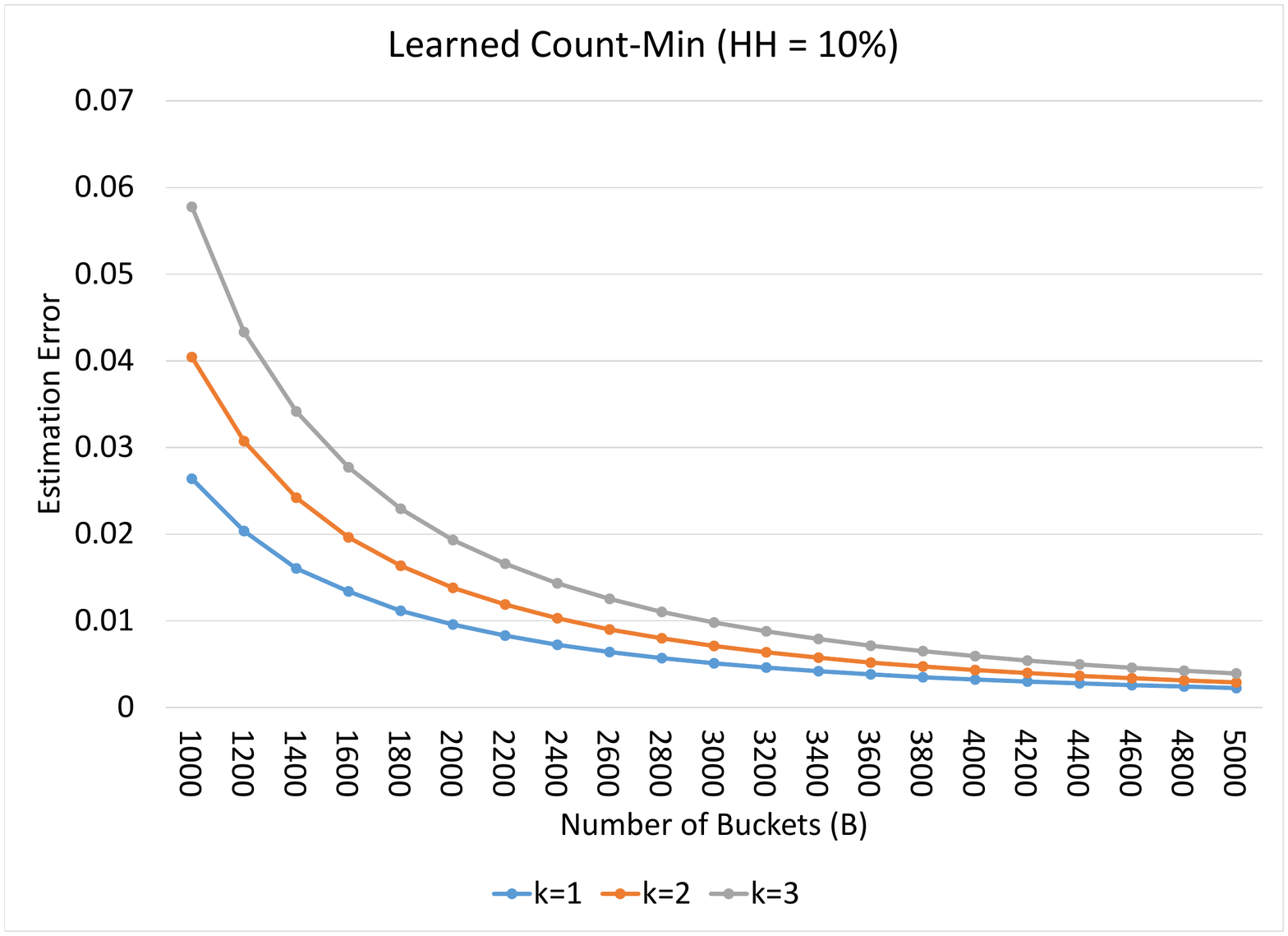}
\endminipage
\caption{The performance of (Learned) Count-Min with different number of rows.}
\label{fig:cm}
\end{figure*}
\begin{figure*}[!h]
\minipage{0.5\textwidth}
		\includegraphics[width=\textwidth]{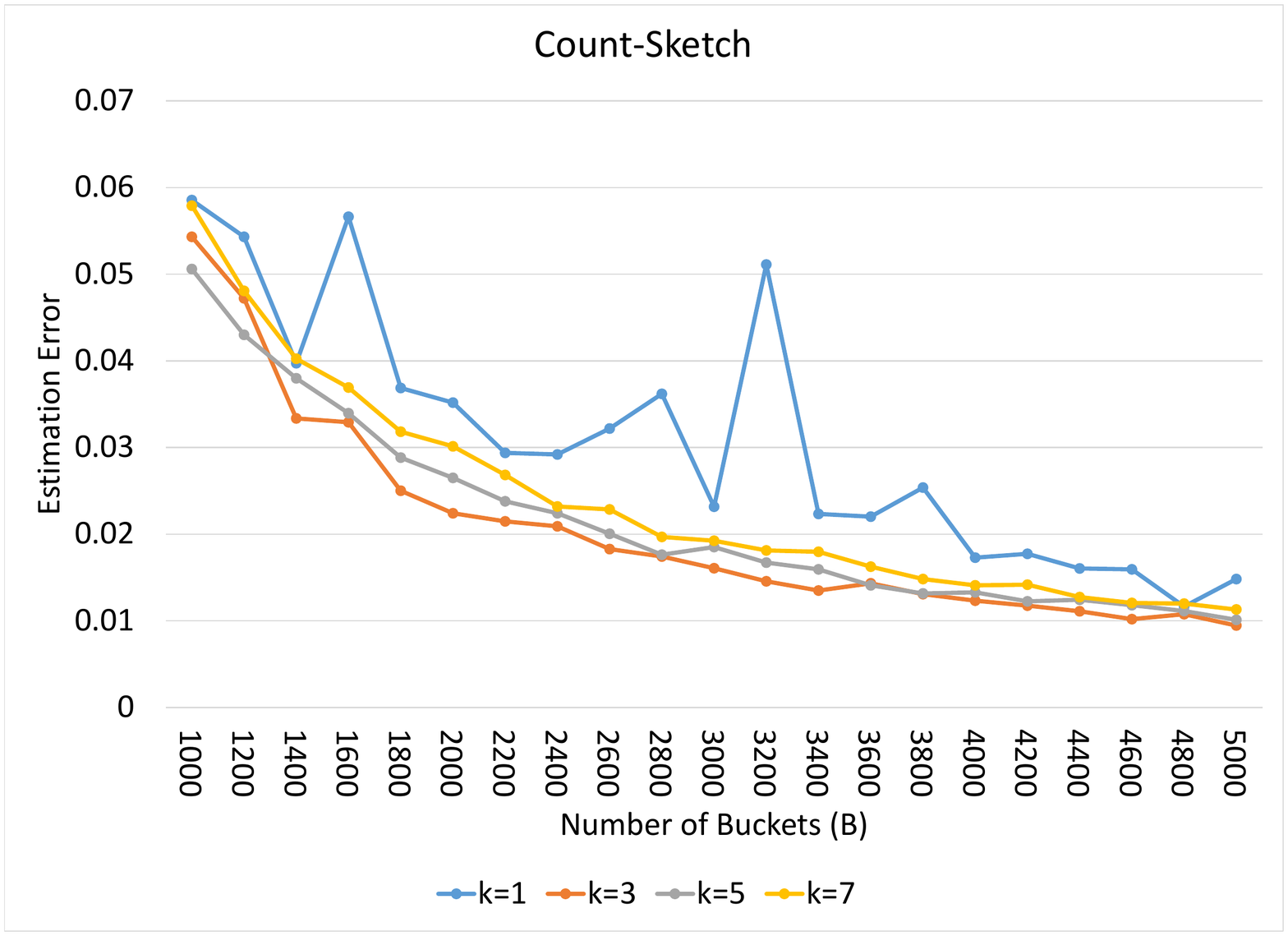}
\endminipage\hfill
\minipage{0.5\textwidth}
		\includegraphics[width=\textwidth]{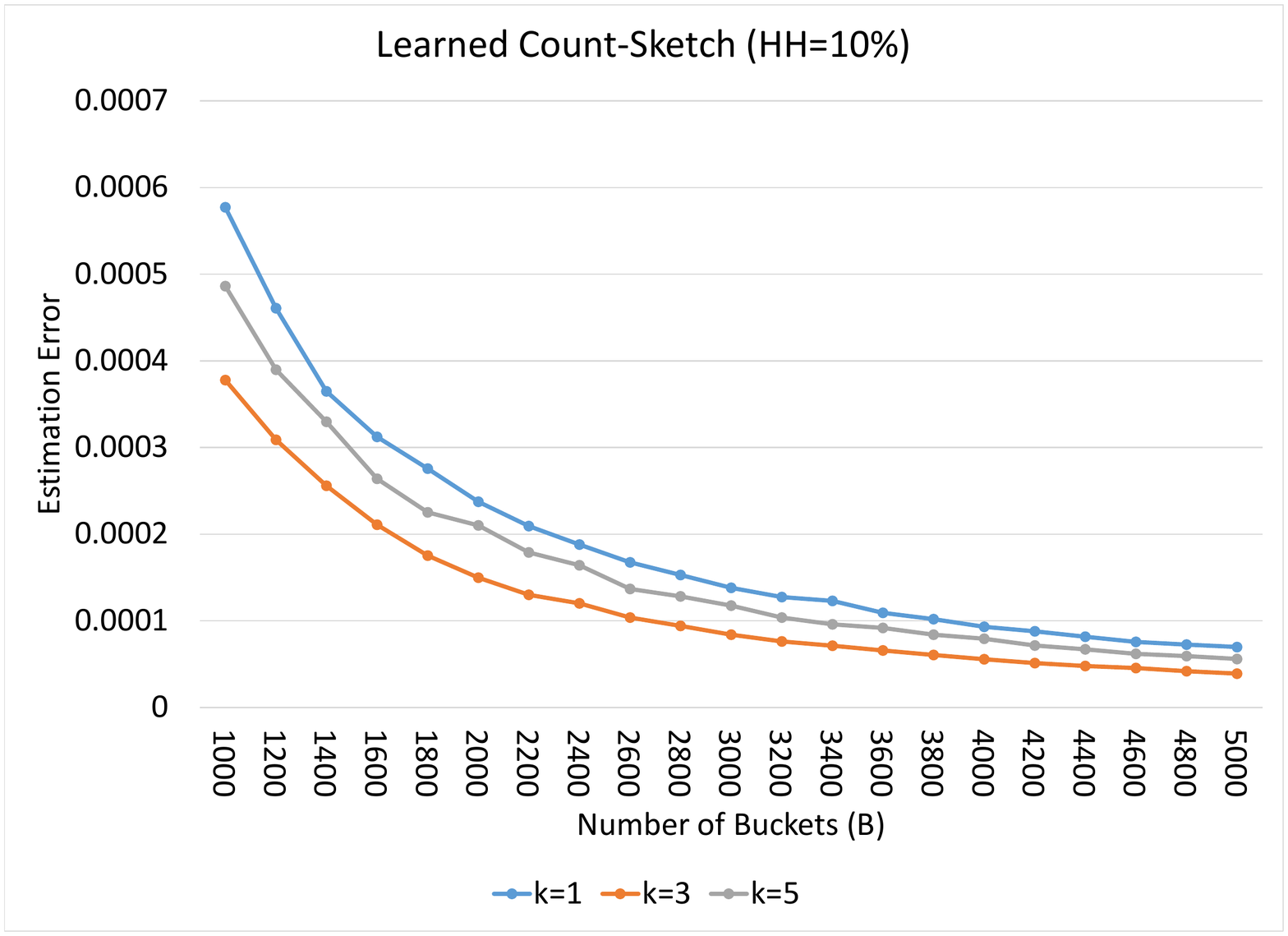}
\endminipage
\caption{The performance of (Learned) Count-Sketch with different number of rows.}
\label{fig:cs}
\end{figure*}

\subparagraph*{Sketches with the same number of buckets but different shapes.} 
Here, we compare the empirical performances of both standard and learned variants of Count-Min and Count-Sketch with varying choices for the parameter. More precisely, we fix the sketch size and vary the number of rows (i.e., number of hash functions) in the sketch.  

As predicted in our theoretical analysis, Figures~\ref{fig:cm} and~\ref{fig:cs} show that setting the number of rows to some constant larger than $1$ for standard CM and CS, leads to a smaller estimation error as we increase the size of the sketch. In contrast, in the learned variant, the average estimation error increases in $k$ being smallest for $k=1$, as was also predicted by our analysis.

\subparagraph*{Learned vs. Standard Sketches.} 
\begin{table}[!t]
\centering	
\resizebox{.95\textwidth}{!}{%
\renewcommand{\arraystretch}{.85}
\begin{tabular}{@{}l|lll|lll@{}}
\toprule
\textsf{\textbf{B}} & \textsf{\textbf{CM ($\boldsymbol{k=1}$)}} & \textsf{\textbf{CM ($\boldsymbol{k=2}$)}} & \textsf{\textbf{L-CM}} & \textsf{\textbf{CS ($\boldsymbol{k=1}$)}} & \textsf{\textbf{CS ($\boldsymbol{k=3}$)}} & \textsf{\textbf{L-CS}} \\ \midrule
1000 & 0.085934 & 0.080569 & \cellcolor[HTML]{9AFF99}0.026391 & 0.058545 & 0.054315 & \cellcolor[HTML]{9AFF99}0.000577138 \\
1200 & 0.077913 & 0.06266  & \cellcolor[HTML]{9AFF99}0.020361 & 0.054322 & 0.047214 & \cellcolor[HTML]{9AFF99}0.000460688 \\
1400 & 0.074504 & 0.052464 & \cellcolor[HTML]{9AFF99}0.016036 & 0.03972  & 0.033348 & \cellcolor[HTML]{9AFF99}0.00036492  \\
1600 & 0.071528 & 0.043798 & \cellcolor[HTML]{9AFF99}0.01338  & 0.056626 & 0.032925 & \cellcolor[HTML]{9AFF99}0.000312238 \\
1800 & 0.059898 & 0.038554 & \cellcolor[HTML]{9AFF99}0.011142 & 0.036881 & 0.025003 & \cellcolor[HTML]{9AFF99}0.000275648 \\
2000 & 0.046389 & 0.033746 & \cellcolor[HTML]{9AFF99}0.009556 & 0.035172 & 0.022403 & \cellcolor[HTML]{9AFF99}0.000237371 \\
2200 & 0.036082 & 0.029059 & \cellcolor[HTML]{9AFF99}0.008302 & 0.029388 & 0.02148  & \cellcolor[HTML]{9AFF99}0.000209376 \\
2400 & 0.032987 & 0.025135 & \cellcolor[HTML]{9AFF99}0.007237 & 0.02919  & 0.020913 & \cellcolor[HTML]{9AFF99}0.00018811  \\
2600 & 0.041896 & 0.023157 & \cellcolor[HTML]{9AFF99}0.006399 & 0.032195 & 0.018271 & \cellcolor[HTML]{9AFF99}0.00016743  \\
2800 & 0.026351 & 0.021402 & \cellcolor[HTML]{9AFF99}0.005694 & 0.036197 & 0.017431 & \cellcolor[HTML]{9AFF99}0.000152933 \\
3000 & 0.032624 & 0.020155 & \cellcolor[HTML]{9AFF99}0.005101 & 0.023175 & 0.016068 & \cellcolor[HTML]{9AFF99}0.000138081 \\
3200 & 0.023614 & 0.018832 & \cellcolor[HTML]{9AFF99}0.004599 & 0.051132 & 0.01455  & \cellcolor[HTML]{9AFF99}0.000127445 \\
3400 & 0.021151 & 0.016769 & \cellcolor[HTML]{9AFF99}0.004196 & 0.022333 & 0.013503 & \cellcolor[HTML]{9AFF99}0.000122947 \\
3600 & 0.021314 & 0.015429 & \cellcolor[HTML]{9AFF99}0.003823 & 0.022012 & 0.014316 & \cellcolor[HTML]{9AFF99}0.000109171 \\
3800 & 0.027798 & 0.014677 & \cellcolor[HTML]{9AFF99}0.003496 & 0.025378 & 0.013082 & \cellcolor[HTML]{9AFF99}0.000102035 \\
4000 & 0.021407 & 0.013279 & \cellcolor[HTML]{9AFF99}0.00322  & 0.017303 & 0.012312 & \cellcolor[HTML]{9AFF99}0.0000931   \\
4200 & 0.020883 & 0.012419 & \cellcolor[HTML]{9AFF99}0.002985 & 0.017719 & 0.011748 & \cellcolor[HTML]{9AFF99}0.0000878   \\
4400 & 0.022383 & 0.011608 & \cellcolor[HTML]{9AFF99}0.002769 & 0.016037 & 0.011097 & \cellcolor[HTML]{9AFF99}0.0000817   \\
4600 & 0.020378 & 0.011151 & \cellcolor[HTML]{9AFF99}0.002561 & 0.015941 & 0.010202 & \cellcolor[HTML]{9AFF99}0.0000757   \\
4800 & 0.015114 & 0.010612 & \cellcolor[HTML]{9AFF99}0.002406 & 0.011642 & 0.010757 & \cellcolor[HTML]{9AFF99}0.0000725   \\
5000 & 0.01603  & 0.009767 & \cellcolor[HTML]{9AFF99}0.002233 & 0.014829 & 0.009451 & \cellcolor[HTML]{9AFF99}0.0000698   \\ \bottomrule
\end{tabular}
}\caption{The estimation error of different sketching methods under Zipfian distribution. In this example, the number of unique items $n$ is equal to $10K$. In the learned variants, number of rows, $k$, is equal to $1$ and the {\em perfect} heavy hitter oracles detect top $c$-frequent items where $c=B/10$.}\label{tbl:experiments}
\end{table}
\begin{figure*}[!h]
\minipage{0.5\textwidth}
		\includegraphics[width=\textwidth]{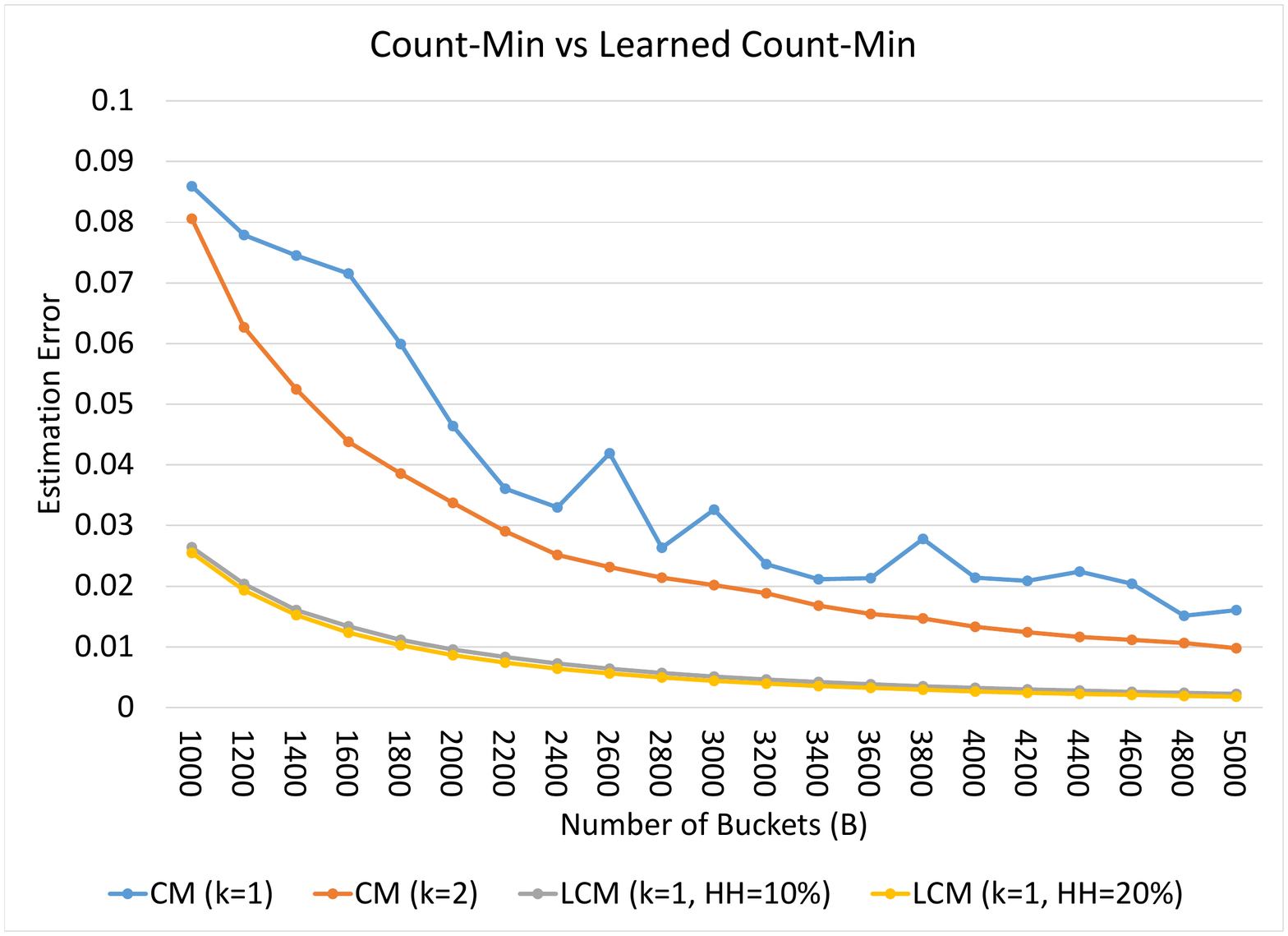}
\endminipage\hfill
\minipage{0.5\textwidth}
		\includegraphics[width=\textwidth]{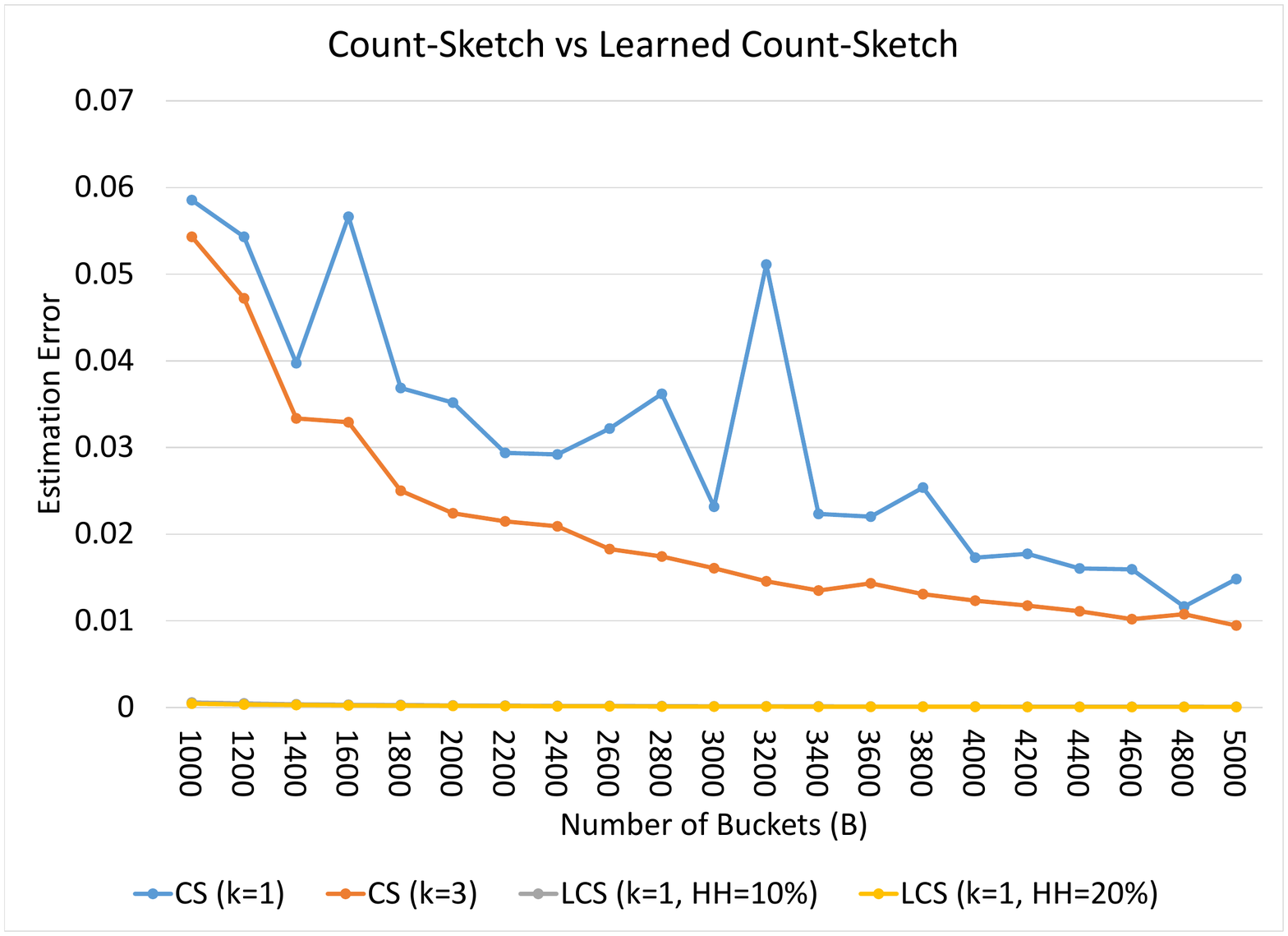}
\endminipage
\caption{The comparison of the performance of learned and standard variants of Count-Min and Count-Sketch. }
\label{fig:learned-standard}
\end{figure*}

In Figure~\ref{fig:learned-standard}, we compare the performance of learned variants of Count-Min and Count-Sketch with the standard Count-Min and Count-Sketch. To be fair, we assume that each bucket that is assigned a heavy hitter consumes two bucket of memory: one for counting the number of times the heavy item appears in the stream and one for indexing the heavy item in the data structure.     

We observe that the learned variants of Count-Min and Count-Sketch significantly improve upon the estimation error of their standard ``non-learned'' variants.  We note that the estimation errors for the learned Count-Sketches in Figure~\ref{fig:learned-standard} are not zero but very close to zero; see Table~\ref{tbl:experiments} for the actual values.

%% file: appendix.tex

\section{Count-Min for General Zipfian (with $\alpha\neq 1$)}\label{appendix}
In this appendix we provide an analysis of the expected error with Count-Min in the case with input coming from a general Zipfian distribution, i.e., $f_i \propto \frac{1}{i^{\alpha}}$, for some fixed $\alpha>0$. By scaling we  can assume that $f_i = \frac{1}{i^{\alpha}}$ with no loss of generality. Our results on the expected error is presented in~\Cref{tbl:results3} below. We start by analyzing the standard Count-Min sketch that does not have access to a machine learning oracle.
\begin{table}[h!]
\centering	
\resizebox{.85\textwidth}{!}{%
\renewcommand{\arraystretch}{1.5}
\begin{tabular}{l|l l} 
\toprule
& $k=1$ & $k > 1$\\
\midrule
\textsf{\textbf{CM, $\alpha<1$}} & $\Theta \left( \frac{n^{2-2\alpha}}{B} \right)$	&	$\Theta \left({k{n^{2-2\alpha}}\over B} \right)$ \\
\textsf{\textbf{CM, $\alpha>1$}} & $O\left( \frac{1}{B}\right)$	&	$(O(1))^k (\log (B))^{k/\alpha+1} \cdot \left( \frac{k^k}{B^k}+\frac{k^\alpha}{B^\alpha}\right) $ and $\Omega \left( \frac{k^k}{B^k} + \frac{k^\alpha}{(B \log k)^\alpha}\right)$\\
\textsf{\textbf{L-CM, $\alpha<1$}} & $\Theta \left( \frac{n^{2-2\alpha}}{B} \right)$ & $\Omega \left( \frac{n^{2-2\alpha}}{B} \right)$\\ 
\textsf{\textbf{L-CM, $\alpha>1$}} & $\Theta \left( B^{1-2\alpha} \right)$ & $\Omega \left( B^{1-2\alpha} \right)$ \\ 
\bottomrule
\end{tabular}
}
\caption{The (scaled) expected errors $\Err(\sF, \tilde{\sF}_{\sA}) =\sum_{i \in [n]} f_i |f_i - \tilde{f}_i|$ of classic and learned Count-Min with $k$ hash functions when the input has a Zipfian distribution with exponent  $\alpha\neq1$.  
The expected errors can be found by normalizing with $\sum_{i \in [n]} f_i$ which is $\Theta(n^{1-\alpha})$ for $\alpha<1$ and $\Theta(1)$ for $\alpha>1$. We note that when $k>1$ is a constant, the upper and lower bounds for CM for $\alpha>1$ are within logarithmic factors of each other. In particular we obtain the combined bound of $\tilde\Theta\left( \frac{1}{B^k}+\frac{1}{B^k} \right)$ in this case, demonstrating that the bounds, even if they appear complicated, are almost tight.}\label{tbl:results3}
\end{table}

\subsection{Standard Count-Min}
We begin by considering the case $\alpha<1$, in which case we have the following result.
\begin{theorem}\label{thm:CSsmallalpha}
Let $0<\alpha <1$ be fixed and $f_i=1/i^\alpha$ for $i \in [n]$. Let $n,B,k \in \N$ with $k\geq 1$ and $B\leq n/k$. Let further $h_1,\dots,h_k: [n] \to [B]$ be independent  and truly random hash functions. For $i \in [n]$ define the random variable $\tilde{f_i}=\min_{\ell \in [k]} \left( \sum_{j\in [n]} [h_{\ell}(j)=h_\ell(i)]f_j \right)$. For any $i\in [n]$ it holds that $\E[|\tilde f_i- f_i|]=\Theta \left( \frac{n^{1-\alpha}}{B} \right)$.
\end{theorem}
We again note the phenomenon that with a \emph{total} of $B$ buckets, i.e., replacing $B$ by $B/k$ in the theorem, the expected error is $\Theta\left( \frac{kn^{1-\alpha}}{B} \right)$, which only increases as we use more hash functions. 
\begin{proof}
For a fixed $\ell \in [k]$ we have that
\begin{align*}
\E\left[\sum_{j\in [n]\setminus \{i\}} [h_{\ell}(j)=h_\ell(i)]f_j \right]=\frac{1}{B}\sum_{j \in[n]\setminus \{i\}}\frac{1}{j^\alpha}=O \left( \frac{n^{1-\alpha}}{B} \right),
\end{align*}
and so $\E[|\tilde f_i- f_i|]=O \left( \frac{n^{1-\alpha}}{B} \right)$. 

For the lower bound, we define $N=[n]\setminus ([B] \cup \{i\})$ and for $\ell \in [k]$, $X_\ell=\sum_{j \in N}[h_{\ell}(j)=h_\ell(i)]f_j$. Simple calculations yield that $\E[X_\ell]=\Theta \left( \frac{n^{1-\alpha}}{B} \right)$ and 
\begin{align*}
\Var[X_\ell]= \begin{cases}
\Theta \left( \frac{\log \left( \frac{n}{B} \right)}{B} \right), & \alpha=1/2, \\
\Theta \left(\frac{n^{1-2\alpha}}{B} \right), & \alpha <1/2, \\
\Theta \left(B^{-2\alpha} \right), & \alpha>1/2.
\end{cases}
\end{align*}
Using Bennett's inequality (\Cref{thm:Bennett}), with $M=B^{-\alpha}$ we obtain that 
\begin{align*}
\Pr[X_\ell\leq \E[X_{\ell}]/2]\leq \begin{cases}
\exp \left(- \Omega(\log (n/B) h(\left(\frac{n}{B} \right)^{1/2} \frac{1}{\log (n/B)})) \right), & \alpha=1/2, \\
\exp \left(- \Omega\left(  \left( \frac{n}{B} \right)^{1-2\alpha} h\left( \left( \frac{n}{B} \right)^{\alpha}\right)\right) \right), & \alpha <1/2, \\
\exp \left(- \Omega\left( h\left( \left( \frac{n}{B} \right)^{1-\alpha}\right)\right) \right), & \alpha>1/2.
\end{cases}
\end{align*}
Using that $n\geq k B$ and~\Cref{asymptotics} we in either case obtain that 
\begin{align*}
\Pr[X_\ell\leq \E[X_{\ell}]/2]=\exp\left(-\Omega( k^{1-\alpha} \log k) \right)=k^{-\Omega(k^{1-\alpha})}.
\end{align*}
As the events $(X_\ell> \E[X_{\ell}]/2)_{\ell \in [k]}$ are independent, they happen simultaneously with probability $(1-k^{-\Omega(k^{1-\alpha})})^k=\Omega(1)$. If they all occur, then $|\tilde f_i- f_i|=\Omega \left( \frac{n^{1-\alpha}}{B} \right)$, so it follows that $\E[|\tilde f_i- f_i|]=\Omega \left( \frac{n^{1-\alpha}}{B} \right)$, as desired.
\end{proof}
Next, we consider the case $\alpha>1$. In this case we have the following theorem where we obtain the result presented in~\Cref{tbl:results3} by replacing $B$ with $B/k$.
\begin{theorem}\label{thm:CSexponent}
Let $\alpha >1$ be fixed and $f_i=1/i^\alpha$ for $i \in [n]$. Let $n,B,k \in \N$ with $k\geq 2$ and $B\leq n/k$. Let further $h_1,\dots,h_k: [n] \to [B]$ be independent  and truly random hash functions. For $i \in [n]$ define the random variable $\tilde{f_i}=\min_{\ell \in [k]} \left( \sum_{j\in [n]} [h_{\ell}(j)=h_\ell(i)]f_j \right)$. For any $i\in [n]$ it holds that
\begin{align*}
\E[|\tilde f_i- f_i|]\leq C^k(\log (B))^{k/\alpha+1} \cdot \left( \frac{1}{B^k}+\frac{1}{B^\alpha}\right),
\end{align*}
for some constant $C$ depending only on $\alpha$. Furthermore, $\E[|\tilde f_i- f_i|]=\Omega \left( \frac{1}{B^k} + \frac{1}{(B \log k)^\alpha}\right)$.
\end{theorem}
\begin{proof}
Let us start by proving the lower bound. Let $N=[\lfloor B\log k \rfloor]$. With probability 
$$
\left(1- \left(1-1/B \right)^{|N\setminus \{i\}|} \right)^k\geq \left(1-e^{\frac{|N\setminus \{i\}|}{B}} \right)^k=\Omega(1)
$$
 it holds that for each $\ell\in[k]$ there exists $j\in N\setminus\{i\}$ such that $h_{\ell}(j)=h_{\ell}(i)$. In this case $|\tilde f_i- f_i|\geq \frac{1}{(B \log k)^\alpha}$, so it follows that also $\E[|\tilde f_i- f_i|]\geq \frac{1}{(B \log k)^\alpha}$.
Note next that with probability $1/B^k$, $h_\ell(1)=h_{\ell}(i)$ for each $\ell \in [k]$. If this happens, $|\tilde f_i- f_i|\geq 1$, so it follows that $\E[|\tilde f_i- f_i|]\geq 1/B^k$ which is the second part of the lower bound.

Next we prove the upper bound. The technique is very similar to the proof of~\Cref{thm:simplecm}. We define $N_1=[B]\setminus \{i\}$ and $N_2=[n] \setminus ([B] \cup \{i\})$. We further define $X_1^{(\ell)}=\sum_{j\in N_1} [h_{\ell}(j)=h_\ell(i)]f_j$ and $X_2^{(\ell)}=\sum_{j\in N_2} [h_{\ell}(j)=h_\ell(i)]f_j$ for $\ell \in [k]$. 
Note that for any $\ell\in [k]$, $\E[X_2^{(\ell)}]=O\left(\frac{1}{B^\alpha} \right)$, so it suffices to bound $\E[\min_{\ell \in [k]}(X_1^{(\ell)})]$. Let $t\geq 3/B^\alpha$ be given. A similar union bound to that given in the proof of~\Cref{thm:simplecm} gives that for any $s\leq t$,
\begin{align*}
\Pr[X_1^{(\ell)}\geq t]\leq \binom{B}{t/s}\frac{1}{B^{t/s}}+\frac{1}{Bs^{1/\alpha}}\leq \left( \frac{es}{t} \right)^{t/s}+ \frac{(t/s)^{1/\alpha}}{Bt^{1/\alpha}}.
\end{align*}
Choosing $s$ such that $t/s=\Theta(\log (Bt^{1/\alpha}))$ is an integer, we obtain the bound
\begin{align*}
\Pr[X_1^{(\ell)}\geq t]  \leq C_1 \frac{(\log (Bt^{1/\alpha}))^{1/\alpha}}{Bt^{1/\alpha}}=C_1 \frac{(\log (Bt^{\gamma}))^{\gamma}}{Bt^{\gamma}},
\end{align*}
where we have put $\gamma=1/\alpha$ and $C_1$ is a universal constant.
Let $Z=\min_{\ell \in [k]}(X_1^{(\ell)})$. Note that $Z\leq \sum_{j=1}^\infty 1/j^{\alpha}\leq C_2$, where $C_2$ is a constant only depending on $\alpha$. Thus
\begin{align*}
\E[Z]&\leq \frac{3}{B^\alpha}+\int_{3/B^\alpha}^{C_2} \Pr[Z \geq t]\, dt\leq \frac{3}{B^\alpha}+\int_{3/B^\alpha}^{C_2} \left(C_1 \frac{(\log (Bt^{\gamma}))^{\gamma}}{Bt^{\gamma}} \right)^k\, dt \\
& \leq \frac{3}{B^\alpha}+\frac{C_3^k \log(B)^{k/\alpha}}{B^k}
\int_{3/B^\alpha}^{C_2} \frac{1}{t^{k/\alpha}}\, dt
\end{align*}
for some constant $C_3$ (depending on $\alpha$). If $k\leq \alpha$, the integral is $O(\log B)$ and this bound suffices. If $k>\alpha$, the integral is $O(B^{k-\alpha})$, which again suffices to give the desired bound.
\end{proof}

\begin{remark}
As discussed in~\Cref{remark:lowindependence} we only require the hash functions to be $O(\log B)$-independent in the proof of the upper bound of~\Cref{thm:CSexponent}. In the upper bound of~\Cref{thm:CSsmallalpha} we only require the hash functions to be $2$-independent.
\end{remark}

\subsection{Learned Count-Min}
We now proceed to analyse the learned Count-Min algorithm which has access to an oracle which, given an item, predicts whether it is among the $B$ heaviest items. The algorithm stores the frequencies of the $B$ heaviest items in $B$ individual buckets, always outputting the exact frequency when queried one of these items. On the remaining items it performs a regular Count-Min sketch with a single hash function hashing to $B$ buckets. 
\begin{theorem}\label{thm:LCMgenZipf}
Let $\alpha >0$ be fixed and $f_i=1/i^\alpha$ for $i \in [n]$. Let $n,B\in \N$ with $2B\leq n$ and $h: [n] \to [B]$ be a $2$-independent hash functions. For $i \in [n]$ define the random variable $\tilde{f_i}= \sum_{j\in [n]\setminus [B]} [h(j)=h(i)]f_j $. Then
\begin{align*}
\E[|\tilde f_i- f_i|]=
\begin{cases}
\Theta \left( \frac{n^{1-\alpha}}{B} \right), & \alpha<1 \\
\Theta \left( B^{-\alpha} \right), & \alpha>1.
\end{cases}
\end{align*}
\end{theorem}
\begin{proof}
Both results follows using linearity of expectation. 
\begin{align*}
\E[|\tilde f_i- f_i|]=\frac{1}{B} \sum_{j\in [n]\setminus( [B] \cup \{i\})}\frac{1}{j^\alpha}=\begin{cases}
\Theta \left( \frac{n^{1-\alpha}}{B} \right), & \alpha <1, \\
\Theta \left( B^{-\alpha} \right), & \alpha >1.
\end{cases}
\end{align*}
\end{proof}
\begin{corollary}\label{cor:learnedgeneralzipf}
Using the learned Count-Min on input coming from a Zipfian distribution with exponent $\alpha$, it holds that 
\begin{align*}
\E\left[\sum_{i \in [n]} f_i \cdot |\tilde f_i- f_i|\right]=
\begin{cases}
\Theta \left( \frac{n^{2-2\alpha}}{B} \right), & \alpha <1, \\
\Theta \left( B^{1-2\alpha} \right), & \alpha >1.
\end{cases}
\end{align*}
\end{corollary}

Why are we only analysing learned Count-Min with a single hash function? After all, might it not be conceivable that more hash functions can reduce the expected error? It turns out that if our aim is to minimize the expected error $\Err(\sF, \tilde{\sF}_{\sA})$ we cannot do better than in~\Cref{cor:learnedgeneralzipf}. Indeed, we can employ similar techniques to those used in~\cite{hsu2018learningbased} to prove the following lower bound extending their result to general exponents $\alpha \neq 1$.

\begin{theorem}\label{thm:CMlowerbound}
Let $\alpha >0$ be fixed and $f_i=1/i^\alpha$ for $i \in [n]$. Let $n,B\in \N$ with $n\geq c B$ for some sufficiently large constant $c$ and let $h: [n] \to [B]$ be \textbf{any} function. For $i \in [n]$ define the random variable $\tilde{f_i}= \sum_{j\in [n]} [h(j)=h(i)]f_j $. Then
\begin{align*}
\sum_{i \in [n]}f_i \cdot |\tilde f_i- f_i|]=
\begin{cases}
\Omega \left( \frac{n^{2-2\alpha}}{B} \right), & \alpha<1 \\
\Omega \left( B^{1-2\alpha} \right), & \alpha>1.
\end{cases}
\end{align*}
\end{theorem}
A simple reduction shows that Count-Min with a total of $B$ buckets and any number of hash functions cannot provide and expected error that is lower than the lower bound in~\Cref{thm:CMlowerbound} (see~\cite{hsu2018learningbased}).
\begin{proof}
We subdivide the exposition into the cases $0<\alpha<1$ and $\alpha>1$.
\paragraph{Case 1: $0<\alpha <1$.}
In this case
\begin{align}\label{eq:squaresofbuckets}
\sum_{i \in [n]} f_i \cdot |\tilde f_i- f_i|\geq \sum_{i \in [n]\setminus [B]} f_i \cdot |\tilde f_i- f_i|
&=\sum_{b \in [B]}\left(\sum_{j\in[n]\setminus[B]: h(j)=b} f_j\right)^2-\sum_{i \in[n]\setminus [B]}f_i^2 \nonumber \\
&=\sum_{b \in [B]}S_b^2-\sum_{i \in[n]\setminus [B]}f_i^2,
\end{align}
where we have put $S_b=\sum_{j\in[n]\setminus [B]: h(j)=b}f_j$, the total weight of items hashing to bucket $b$. Now by Jensen's inequality
\begin{align*}
\sum_{b \in [B]}S_b^2\geq\frac{1}{B} \left(\sum_{i \in [n]\setminus [B]}f_i\right)^2
\end{align*}
Furthermore, we have the estimates
\begin{align*}
\sum_{i \in [n]\setminus [B]}f_i=\sum_{i=B}^n \frac{1}{i^\alpha}-\frac{1}{B^\alpha}  \geq \int_{B}^n x^{-\alpha} \, dx-\frac{1}{B^\alpha} =\frac{1}{1-\alpha}(n^{1-\alpha}-B^{1-\alpha})-\frac{1}{B^\alpha},
\end{align*}
and
\begin{align*}
\sum_{i \in [n] \setminus [B]} f_i^2\leq \int_{B}^n x^{-2\alpha}= \begin{cases}
\frac{1}{1-2\alpha}(n^{1-2\alpha}-B^{1-2\alpha}), &\alpha\neq 1/2 \\
\log (n/B), & \alpha=1/2.
\end{cases}
\end{align*}
Here we have used the standard technique of comparing a sum to an integral.
Assuming that $n\geq cB$ for some sufficiently large constant $c$ (depending on $\alpha$), it follows that $\sum_{b \in [B]}S_b^2=\Omega \left(\frac{n^{2-2\alpha}}{B} \right)$. It moreover follows (again for $n$ sufficiently large) that,
\begin{align*}
\sum_{i \in[n]\setminus [B]}f_i^2= \begin{cases}
O(\log (n/B)), & \alpha=1/2, \\
O(n^{1-2 \alpha}), & \alpha <1/2, \\
O(B^{1-2\alpha}), & \alpha>1/2.
\end{cases}
\end{align*}
Plugging into~\eqref{eq:squaresofbuckets}, we see that in each of the three cases $\alpha <1/2$, $\alpha=1/2$ and $\alpha>1/2$ it holds that $\sum_{b \in [B]}S_b^2-\sum_{i \in[n]}f_i^2= \Omega \left(\frac{n^{2-2\alpha}}{B}\right)$.

\paragraph{Case 2: $0\alpha>1$.} For this case we simply assume that $n\geq 3B$. Let $I\subseteq [3B] \setminus [B]$ consist of those $i$ satisfying that $h(i)=h(j)$ for some $j\in [3B] \setminus [B]$, $j\neq i$. Then $|I| \geq B$ and if $i\in I$, then $f_i\geq (3B)^{-\alpha}$ and $|\tilde f_i-f_i|\geq (3B)^{-\alpha}$. Thus
\begin{align*}
\sum_{i \in [n]} f_i \cdot |\tilde f_i- f_i|\geq\sum_{i \in I}f_i \cdot |\tilde f_i- f_i| \geq B (3B)^{-2\alpha}=\Omega(B^{1-2\alpha}).
\end{align*}
\end{proof}

\subsection{Learned Count-Min using a noisy oracle}
As we did in the case $\alpha=1$ (\Cref{thm:faultyoracle}), we now present an analogue to~\Cref{thm:LCMgenZipf} when the heavy hitter oracle is noisy. Note that the results in~\Cref{tbl:results3} demonstrates that we obtain no asymptotic improvement using the heavy hitter oracle when $0<\alpha<1$ and therefore we only consider the case $\alpha>1$. We show the following trade-off between the classic and learned case, as the error probability, $\delta$, that the heavy hitter oracle misclassifies an item, varies in $[0,1]$.

\begin{theorem}\label{thm:faultyoracle2}
Suppose that the input follows a generalized Zipfian distribution with $n\geq B$ different items and exponent $\alpha$ for some constant $\alpha>1$. Learned Count-Sketch with a single hash functions, a heavy hitter oracle $\textbf{HH}_\delta$, $B_h=\Theta(B)$ bins allocated to the $B_h$ items classified as heavy and $B-B_h=\Theta(B)$ bins allocated to a Count-Sketch of the remaining items, incurs an expected error of 
$$
O\left(\frac{1}{B} \left(\delta+B^{1-\alpha} \right)^2\right)
$$
\end{theorem}
\begin{proof}
The proof is very similar to that of~\Cref{thm:faultyoracle}
Let $h:[n] \to [B-B_h]$ be the hash function used for the Count-Min. In the analysis to follow, it is enough to assume that it is $2$-independent. Suppose item $i$ is classified as non-heavy. The expected error incurred by item $i$ is then
$$
\frac{1}{B-B_h} \left(\delta\sum_{j \in [B_h] \setminus \{i\}} f_j+\sum_{j \in [n] \setminus ([B_h]\cup \{i\})} f_j \right)=O\left( \frac{\delta +B^{1-\alpha}}{B} \right).
$$
Letting $N=\sum_{i \in [n]}f_i=O(1)$, the expected error (as defined in~\eqref{eq:simplified-error}) is at most
$$
\frac{1}{N}\left(\delta\sum_{j \in [B_h] \setminus \{i\}} f_j+\sum_{j \in [n] \setminus ([B_h]\cup \{i\})} f_j \right)\cdot O\left( \frac{\delta +B^{1-\alpha}}{B} \right)=O\left(\frac{1}{B} \left(\delta+B^{1-\alpha} \right)^2\right),
$$
as desired.
\end{proof}
For $\delta=1$, we recover the bound for the classic Count-Min. We also see that it suffices that $\delta=O(B^{1-\alpha})$ in order to obtain the same bound as with a perfect heavy hitter oracle.

%% file: appendix2.tex

\section{Concentration bounds}\label{appendix2}
In this appendix we collect some concentration inequalities for reference in the main body of the paper. The inequality we will use the most is Bennett's inequality. However, we remark that for our applications, several other variance based concentration result would suffice, e.g., Bernstein's inequality. 
 \begin{theorem}[Bennett's inequality~\cite{bennett1962}]\label{thm:Bennett}
Let $X_1,\dots,X_n$ be independent, mean zero random variables. Let $S=\sum_{i=1}^n X_i$, and $\sigma^2,M>0$ be such that $\Var[S]\leq \sigma^2$ and $|X_i|\leq M$ for all $i\in [n]$. For any $t\geq 0$,
\begin{align*}
\Pr[S\geq t]\leq \exp \left(-\frac{\sigma^2}{M^2} h \left( \frac{tM}{\sigma^2} \right)\right),
\end{align*}
where $h:\R_{\geq 0} \to \R_{\geq 0 }$ is defined by $h(x)=(x+1) \log (x+1)-x$. The same tail bound holds on the probability $\Pr[S\leq -t]$.
\end{theorem}
\remark\label{asymptotics}For $x\geq 0$, $\frac{1}{2}x \log (x+1) \leq h(x) \leq x \log (x+1)$. We will use these asymptotic bounds repeatedly in this paper. 

A corollary of Bennett's inequality is the classic Chernoff bounds.

\begin{theorem}[Chernoff~\cite{Che52:chernoff}]\label{thm:Chernoff}
Let $X_1,\dots,X_n\in [0,1]$ be independent random variables and $S=\sum_{i=1}^n X_i$. Let $\mu=\E[S]$. Then
\begin{align*}
\Pr[S\geq (1+\delta)\mu ] \leq \exp(-\mu h(\delta)).
\end{align*}
\end{theorem}

Even weaker than Chernoff's inequality is Hoeffding's inequality.

\begin{theorem}[Hoeffding~\cite{hoeffding1963inequality}]\label{thm:Hoeffding}
Let $X_1,\dots,X_n \in [0,1]$ be independent random variables. Let $S=\sum_{i=1}^n X_i$. Then 
\begin{align*}
\Pr[S-\E[S]\geq t]\leq e^{-\frac{2t^2}{n}}.
\end{align*}
\end{theorem}

%% file: main.bbl
\newcommand{\etalchar}[1]{$^{#1}$}
\begin{thebibliography}{WLKC16}

\bibitem[ABL{\etalchar{+}}17]{anderson2017high}
Daniel Anderson, Pryce Bevan, Kevin Lang, Edo Liberty, Lee Rhodes, and Justin
  Thaler.
\newblock A high-performance algorithm for identifying frequent items in data
  streams.
\newblock In {\em Proceedings of the 2017 Internet Measurement Conference},
  pages 268--282, 2017.

\bibitem[ACC{\etalchar{+}}11]{ailon2011self}
Nir Ailon, Bernard Chazelle, Kenneth~L Clarkson, Ding Liu, Wolfgang Mulzer, and
  C~Seshadhri.
\newblock Self-improving algorithms.
\newblock {\em SIAM Journal on Computing}, 40(2):350--375, 2011.

\bibitem[ACE{\etalchar{+}}20]{antoniadis2020online}
Antonios Antoniadis, Christian Coester, Marek Elias, Adam Polak, and Bertrand
  Simon.
\newblock Online metric algorithms with untrusted predictions.
\newblock {\em arXiv preprint arXiv:2003.02144}, 2020.

\bibitem[ADJ{\etalchar{+}}20]{angelopoulos2020online}
Spyros Angelopoulos, Christoph D{\"u}rr, Shendan Jin, Shahin Kamali, and Marc
  Renault.
\newblock Online computation with untrusted advice.
\newblock In {\em 11th Innovations in Theoretical Computer Science Conference
  (ITCS 2020)}. Schloss Dagstuhl-Leibniz-Zentrum f{\"u}r Informatik, 2020.

\bibitem[AKL{\etalchar{+}}19]{alabi2019learning}
Daniel Alabi, Adam~Tauman Kalai, Katrina Ligett, Cameron Musco, Christos
  Tzamos, and Ellen Vitercik.
\newblock Learning to prune: Speeding up repeated computations.
\newblock In {\em Conference on Learning Theory}, 2019.

\bibitem[BCI{\etalchar{+}}17]{braverman2017bptree}
Vladimir Braverman, Stephen~R Chestnut, Nikita Ivkin, Jelani Nelson, Zhengyu
  Wang, and David~P Woodruff.
\newblock Bptree: an {$\ell_2$} heavy hitters algorithm using constant memory.
\newblock In {\em Proceedings of the 36th ACM SIGMOD-SIGACT-SIGAI Symposium on
  Principles of Database Systems}, pages 361--376, 2017.

\bibitem[BCIW16]{braverman2016beating}
Vladimir Braverman, Stephen~R Chestnut, Nikita Ivkin, and David~P Woodruff.
\newblock Beating countsketch for heavy hitters in insertion streams.
\newblock In {\em Proceedings of the forty-eighth annual ACM symposium on
  Theory of Computing}, pages 740--753, 2016.

\bibitem[BDSV18]{balcan2018learning}
Maria-Florina Balcan, Travis Dick, Tuomas Sandholm, and Ellen Vitercik.
\newblock Learning to branch.
\newblock In {\em International Conference on Machine Learning}, pages
  353--362, 2018.

\bibitem[BDV18]{balcan2018dispersion}
Maria-Florina Balcan, Travis Dick, and Ellen Vitercik.
\newblock Dispersion for data-driven algorithm design, online learning, and
  private optimization.
\newblock In {\em 2018 IEEE 59th Annual Symposium on Foundations of Computer
  Science (FOCS)}, pages 603--614. IEEE, 2018.

\bibitem[BDW18]{bhattacharyya2018optimal}
Arnab Bhattacharyya, Palash Dey, and David~P Woodruff.
\newblock An optimal algorithm for {$\ell_1$}-heavy hitters in insertion
  streams and related problems.
\newblock {\em ACM Transactions on Algorithms (TALG)}, 15(1):1--27, 2018.

\bibitem[Ben62]{bennett1962}
George Bennett.
\newblock Probability inequalities for the sum of independent random variables.
\newblock {\em Journal of the American Statistical Association},
  57(297):33--45, 1962.

\bibitem[BICS10]{berinde2010space}
Radu Berinde, Piotr Indyk, Graham Cormode, and Martin~J Strauss.
\newblock Space-optimal heavy hitters with strong error bounds.
\newblock {\em ACM Transactions on Database Systems (TODS)}, 35(4):1--28, 2010.

\bibitem[BJPD17]{bora2017compressed}
Ashish Bora, Ajil Jalal, Eric Price, and Alexandros~G Dimakis.
\newblock Compressed sensing using generative models.
\newblock In {\em International Conference on Machine Learning}, pages
  537--546, 2017.

\bibitem[CCFC02]{charikar2002finding}
Moses Charikar, Kevin Chen, and Martin Farach-Colton.
\newblock Finding frequent items in data streams.
\newblock In {\em International Colloquium on Automata, Languages, and
  Programming}, pages 693--703. Springer, 2002.

\bibitem[CGP20]{cohen2020composable}
Edith Cohen, Ofir Geri, and Rasmus Pagh.
\newblock Composable sketches for functions of frequencies: Beyond the worst
  case.
\newblock {\em arXiv preprint arXiv:2004.04772}, 2020.

\bibitem[CGT{\etalchar{+}}19]{chawla2019learning}
Shuchi Chawla, Evangelia Gergatsouli, Yifeng Teng, Christos Tzamos, and Ruimin
  Zhang.
\newblock Learning optimal search algorithms from data.
\newblock {\em arXiv preprint arXiv:1911.01632}, 2019.

\bibitem[CH08]{cormode2008finding}
Graham Cormode and Marios Hadjieleftheriou.
\newblock Finding frequent items in data streams.
\newblock {\em Proceedings of the VLDB Endowment}, 1(2):1530--1541, 2008.

\bibitem[CH10]{cormode2010methods}
Graham Cormode and Marios Hadjieleftheriou.
\newblock Methods for finding frequent items in data streams.
\newblock {\em The VLDB Journal}, 19(1):3--20, 2010.

\bibitem[Che52]{Che52:chernoff}
Herman Chernoff.
\newblock A measure of asymptotic efficiency for tests of a hypothesis based on
  the sum of observations.
\newblock {\em Annals of Mathematical Statistics}, 23(4):493--507, 1952.

\bibitem[CM05a]{cormode2005improved}
Graham Cormode and Shan Muthukrishnan.
\newblock An improved data stream summary: the count-min sketch and its
  applications.
\newblock {\em Journal of Algorithms}, 55(1):58--75, 2005.

\bibitem[CM05b]{cormode2005summarizing}
Graham Cormode and Shan Muthukrishnan.
\newblock Summarizing and mining skewed data streams.
\newblock In {\em Proceedings of the 2005 SIAM International Conference on Data
  Mining}, pages 44--55. SIAM, 2005.

\bibitem[DIRW19]{dong2019learning}
Yihe Dong, Piotr Indyk, Ilya Razenshteyn, and Tal Wagner.
\newblock Learning sublinear-time indexing for nearest neighbor search.
\newblock {\em arXiv preprint arXiv:1901.08544}, 2019.

\bibitem[Erd45]{erdos1945lemma}
Paul Erd{\"o}s.
\newblock On a lemma of littlewood and offord.
\newblock {\em Bulletin of the American Mathematical Society}, 51(12):898--902,
  1945.

\bibitem[GP19]{pmlr-v97-gollapudi19a}
Sreenivas Gollapudi and Debmalya Panigrahi.
\newblock Online algorithms for rent-or-buy with expert advice.
\newblock In {\em Proceedings of the 36th International Conference on Machine
  Learning}, pages 2319--2327, 2019.

\bibitem[GR17]{gupta2017pac}
Rishi Gupta and Tim Roughgarden.
\newblock A pac approach to application-specific algorithm selection.
\newblock {\em SIAM Journal on Computing}, 46(3):992--1017, 2017.

\bibitem[HIKV19]{hsu2018learningbased}
Chen-Yu Hsu, Piotr Indyk, Dina Katabi, and Ali Vakilian.
\newblock Learning-based frequency estimation algorithms.
\newblock In {\em International Conference on Learning Representations}, 2019.

\bibitem[Hoe63]{hoeffding1963inequality}
Wassily Hoeffding.
\newblock Probability inequalities for sums of bounded random variables.
\newblock {\em Journal of the American Statistical Association},
  58(301):13--30, 1963.

\bibitem[IVY19]{indyk2019learning}
Piotr Indyk, Ali Vakilian, and Yang Yuan.
\newblock Learning-based low-rank approximations.
\newblock In {\em Advances in Neural Information Processing Systems}, pages
  7400--7410, 2019.

\bibitem[JLL{\etalchar{+}}20]{jiang2020learningaugmented}
Tanqiu Jiang, Yi~Li, Honghao Lin, Yisong Ruan, and David~P. Woodruff.
\newblock Learning-augmented data stream algorithms.
\newblock In {\em International Conference on Learning Representations}, 2020.

\bibitem[KBC{\etalchar{+}}18]{kraska2017case}
Tim Kraska, Alex Beutel, Ed~H Chi, Jeffrey Dean, and Neoklis Polyzotis.
\newblock The case for learned index structures.
\newblock In {\em Proceedings of the 2018 International Conference on
  Management of Data}, pages 489--504, 2018.

\bibitem[KDZ{\etalchar{+}}17]{khalil2017learning}
Elias Khalil, Hanjun Dai, Yuyu Zhang, Bistra Dilkina, and Le~Song.
\newblock Learning combinatorial optimization algorithms over graphs.
\newblock In {\em Advances in Neural Information Processing Systems}, pages
  6348--6358, 2017.

\bibitem[Kod19]{kodialam2019optimal}
Rohan Kodialam.
\newblock Optimal algorithms for ski rental with soft machine-learned
  predictions.
\newblock {\em arXiv preprint arXiv:1903.00092}, 2019.

\bibitem[LLMV20]{lattanzi2020online}
Silvio Lattanzi, Thomas Lavastida, Benjamin Moseley, and Sergei Vassilvitskii.
\newblock Online scheduling via learned weights.
\newblock In {\em Proceedings of the Fourteenth Annual ACM-SIAM Symposium on
  Discrete Algorithms}, pages 1859--1877. SIAM, 2020.

\bibitem[LNNT16]{larsen2016heavy}
Kasper~Green Larsen, Jelani Nelson, Huy~L Nguy{\^e}n, and Mikkel Thorup.
\newblock Heavy hitters via cluster-preserving clustering.
\newblock In {\em 2016 IEEE 57th Annual Symposium on Foundations of Computer
  Science (FOCS)}, pages 61--70. IEEE, 2016.

\bibitem[LO39]{littlewood1939number}
John~Edensor Littlewood and Albert~C Offord.
\newblock On the number of real roots of a random algebraic equation. ii.
\newblock In {\em Mathematical Proceedings of the Cambridge Philosophical
  Society}, volume~35, pages 133--148. Cambridge University Press, 1939.

\bibitem[LV18]{lykouris2018competitive}
Thodoris Lykouris and Sergei Vassilvitskii.
\newblock Competitive caching with machine learned advice.
\newblock In {\em International Conference on Machine Learning}, pages
  3302--3311, 2018.

\bibitem[M{\etalchar{+}}05]{muthukrishnan2005data}
Shanmugavelayutham Muthukrishnan et~al.
\newblock Data streams: Algorithms and applications.
\newblock {\em Foundations and Trends{\textregistered} in Theoretical Computer
  Science}, 1(2):117--236, 2005.

\bibitem[MAEA05]{metwally2005efficient}
Ahmed Metwally, Divyakant Agrawal, and Amr El~Abbadi.
\newblock Efficient computation of frequent and top-k elements in data streams.
\newblock In {\em International Conference on Database Theory}, pages 398--412.
  Springer, 2005.

\bibitem[MG82]{misra1982finding}
Jayadev Misra and David Gries.
\newblock Finding repeated elements.
\newblock {\em Science of computer programming}, 2(2):143--152, 1982.

\bibitem[Mit18]{mitz2018model}
Michael Mitzenmacher.
\newblock A model for learned bloom filters and optimizing by sandwiching.
\newblock In {\em Advances in Neural Information Processing Systems}, pages
  464--473, 2018.

\bibitem[Mit20]{mitzenmacher2020scheduling}
Michael Mitzenmacher.
\newblock Scheduling with predictions and the price of misprediction.
\newblock In {\em 11th Innovations in Theoretical Computer Science Conference
  (ITCS 2020)}. Schloss Dagstuhl-Leibniz-Zentrum f{\"u}r Informatik, 2020.

\bibitem[MM02]{manku2002approximate}
Gurmeet~Singh Manku and Rajeev Motwani.
\newblock Approximate frequency counts over data streams.
\newblock In {\em VLDB'02: Proceedings of the 28th International Conference on
  Very Large Databases}, pages 346--357. Elsevier, 2002.

\bibitem[MP14]{minton2014improved}
Gregory~T Minton and Eric Price.
\newblock Improved concentration bounds for count-sketch.
\newblock In {\em Proceedings of the twenty-fifth annual ACM-SIAM symposium on
  Discrete algorithms}, pages 669--686. Society for Industrial and Applied
  Mathematics, 2014.

\bibitem[MPB15]{mousavi2015deep}
Ali Mousavi, Ankit~B Patel, and Richard~G Baraniuk.
\newblock A deep learning approach to structured signal recovery.
\newblock In {\em Communication, Control, and Computing (Allerton), 2015 53rd
  Annual Allerton Conference on}, pages 1336--1343. IEEE, 2015.

\bibitem[PSK18]{purohit2018improving}
Manish Purohit, Zoya Svitkina, and Ravi Kumar.
\newblock Improving online algorithms via ml predictions.
\newblock In {\em Advances in Neural Information Processing Systems}, pages
  9661--9670, 2018.

\bibitem[RKA16]{roy2016augmented}
Pratanu Roy, Arijit Khan, and Gustavo Alonso.
\newblock Augmented sketch: Faster and more accurate stream processing.
\newblock In {\em Proceedings of the 2016 International Conference on
  Management of Data}, pages 1449--1463, 2016.

\bibitem[Roh20]{rohatgi2020near}
Dhruv Rohatgi.
\newblock Near-optimal bounds for online caching with machine learned advice.
\newblock In {\em Proceedings of the Fourteenth Annual ACM-SIAM Symposium on
  Discrete Algorithms}, pages 1834--1845. SIAM, 2020.

\bibitem[WLKC16]{wang2016learning}
Jun Wang, Wei Liu, Sanjiv Kumar, and Shih-Fu Chang.
\newblock Learning to hash for indexing big data - a survey.
\newblock {\em Proceedings of the IEEE}, 104(1):34--57, 2016.

\end{thebibliography}
